\documentclass[onefignum,onetabnum]{siamonline181217}

\usepackage[whole]{bxcjkjatype}

\usepackage{lipsum}
\usepackage{amsfonts}
\usepackage{graphicx}
\usepackage{epstopdf}

\usepackage{algorithmic} 
\Crefname{ALC@unique}{Line}{Lines}

\usepackage{amsopn}

\newcommand{\revise}[1]{{\color{black} #1}}
\newcommand{\revisee}[1]{{\color{black} #1}}

\DeclareMathOperator{\diag}{diag}
\DeclareMathOperator*{\argmin}{argmin}
\DeclareMathOperator*{\argmax}{argmax}

\newcommand{\rmd}{\mathrm d}

\newcommand{\rmN}{\mathrm N}
\newcommand{\bbR}{\mathbb R}

\newcommand{\tx}{\tilde{x}}
\newcommand{\tq}{\tilde{q}}
\newcommand{\tp}{\tilde{p}}
\newcommand{\tlambda}{\tilde{\lambda}}
\newcommand{\tnu}{\tilde{\nu}}

\usepackage{mathtools}
\DeclarePairedDelimiter\paren{\lparen}{\rparen}

\ifpdf
  \DeclareGraphicsExtensions{.eps,.pdf,.png,.jpg}
\else
  \DeclareGraphicsExtensions{.eps}
\fi

\usepackage{enumitem}
\setlist[enumerate]{leftmargin=.5in}
\setlist[itemize]{leftmargin=.5in}

\newsiamremark{remark}{Remark}
\newsiamremark{hypothesis}{Hypothesis}
\crefname{hypothesis}{Hypothesis}{Hypotheses}
\newsiamthm{claim}{Claim}
\newsiamremark{example}{Example}

\newsiamremark{assumption}{Assumption}

\usepackage{tikz}
\usepackage{pgfplotstable}
\usetikzlibrary{arrows,positioning,plotmarks,external,patterns,angles,
decorations.pathmorphing,backgrounds,fit,shapes,graphs,calc,spy}
\pgfplotsset{compat=1.14}

\headers{ODE Estimation by quantifying discretization error}{Takeru Matsuda and Yuto Miyatake}

\title{Estimation of ordinary differential equation models \\with discretization error quantification
}

\author{Takeru Matsuda\thanks{Department of Mathematical Informatics, Graduate School of Information Science and Technology, The University of Tokyo, Japan; Mathematical Informatics Collaboration Unit, RIKEN Center for Brain Science, Japan
 (\email{takeru.matsuda@riken.jp})}
\and Yuto Miyatake\thanks{Cybermedia Center, Osaka University, Japan
  (\email{miyatake@cas.cmc.osaka-uac.jp}).}
}

\ifpdf
\hypersetup{
  pdftitle={Estimation of ordinary differential equation models with discretization error quantification},
  pdfauthor={T. Matsuda and Y. Miyatake}
}
\fi

\begin{document}

\maketitle

\begin{abstract}
We consider \revise{parameter} estimation of ordinary differential equation (ODE) models from noisy observations. 
For this problem, one conventional approach is to fit numerical solutions (e.g., Euler, Runge--Kutta) of ODEs to data. 
However, such a method does not account for the discretization error in numerical solutions and has limited estimation accuracy. 
In this study, we develop an estimation method that quantifies the discretization error based on data. 
The key idea is to model the discretization error as random variables and estimate their variance simultaneously with the ODE parameter. 
The proposed method has the form of iteratively reweighted least squares, where the discretization error variance is updated with the isotonic regression algorithm and the ODE parameter is updated by solving a weighted least squares problem using the adjoint system. 
\revise{Experimental results demonstrate that the proposed method attains robust estimation with at least comparable accuracy to the conventional method by successfully quantifying the reliability of the numerical solutions.}
\end{abstract}



\section{Introduction}
\label{sec:1}

A system of ordinary differential equations (ODEs) is a fundamental tool for modeling a dynamical system in many fields.
For example, \revise{spiking} neuron activity is simply described by the FitzHugh--Nagumo model \cite{f61,nay62}:
\begin{align*}
\frac{\rmd V}{\rmd t}=c\paren*{V-\cfrac{V^3}{3}+R}, \quad \frac{\rmd R}{\rmd t}=-\cfrac{1}{c} \paren*{V-a+bR}.
\end{align*}
In practice, ODE models often include unknown system parameters (e.g., $a$, $b$, and $c$ in the above model) or unknown initial state (e.g., $V(0)$ and $R(0)$ in the above model).
In this study, we focus on estimation of ODE models from noisy data.
Specifically, consider an ODE model
\begin{align*}
    \frac{\rmd }{\rmd t} x(t; \theta) = f(x(t; \theta),\theta), \quad x(0; \theta) = x_0(\theta), %\quad t\in[0,T],
\end{align*}
where $\theta$ is an unknown \revise{identifiable} parameter.
Suppose that we have noisy observations $y_1,\dots,y_K$ of $x(t; \theta)$ at several time points $t_1,\dots,t_K$:
\begin{align}
    \label{obs_model} 
    y_k = H x(t_k; \theta) + \varepsilon_k, \quad \varepsilon_k \sim \rmN (0,\Gamma), \quad k = 1,\dots,K,
\end{align}
where $H$ is a given matrix and $\varepsilon_k$ is the observation noise with covariance $\Gamma$.
In this setting, we are interested in estimating $\theta$ based on $y_1,\dots,y_K$.

Estimation of ODE models has a distinct feature due to the intractability of the initial value problems.
Namely, ODEs do not have a closed-form solution in general, and thus, we do not have access to the exact solution $x(t_k; \theta)$ in \cref{obs_model}.
In this sense, our problem here is essentially different from \revise{usual} nonlinear regression.
Instead of the exact solution $x(t_k; \theta)$, its approximation $\tx(t_k; \theta)$ is obtained by using numerical methods for ODEs such as Euler and Runge--Kutta \cite{bu16,hn93}.
Thus, one simple conventional approach to parameter estimation in ODE models is to fit the numerical solutions to data directly.
However, this approach requires numerical solutions to be sufficiently accurate, which is computationally intensive in general.
In other words, the difference $\tx(t_k; \theta)-x(t_k; \theta)$ between the numerical and exact solutions, which we call the discretization error in the following, is not negligible.
\revise{The discretization error is regarded as one of the main uncertainties in parameter estimation of ODE models.}

In this study, we develop a method for estimating ODE models that quantifies the discretization error based on data. 
The key idea is to model the discretization error as Gaussian random variables, which is inspired from the recent studies on probabilistic numerics for differential equations \cite{cg17,hog15,ls19}, and estimate their variance simultaneously  with the ODE parameter. 
\revisee{Note that, although the assumption of discretizaiton error being Gaussian random variables may seem unrealistic, the interest here is whether such a probabilistic model leads to a useful method of parameter estimation or not.
We leave the validation and perhaps modification of the modelling assumption to our future work.}
Specifically, the proposed method has the form of Iteratively Reweighted Least Squares (IRLS), where the discretization error variance and ODE parameter are alternately updated.
The update of the discretization error variance is efficiently solved by the isotonic regression algorithm \cite{bb72,rwd88,vE06}.
On the other hand, the update of the ODE parameter is formulated as a weighted least squares problem and solved by using the exact gradient of the objective function, which is computed by the adjoint system of \cref{obs_model} and symplectic partitioned Runge--Kutta methods~\cite{ss16}.
Experimental results on several ODE models demonstrate that the proposed method attains robust estimation with at least comparable accuracy to the conventional method by successfully quantifying the reliability of the numerical solutions based on data.
For example, the proposed method is beneficial when sufficiently accurate numerical solutions, which may be obtained by the step-size control or validated numerics, are not available due to limited computational resources. 
Note that, whereas several recent studies investigated Bayesian probabilistic numerical methods for ODE models \cite{cc16,cosg19,ocag19}, we focus on non-Bayesian parameter estimation like the maximum likelihood in this study.

This paper is organized as follows.
After presenting the problem setting in \cref{sec:setting}, we explain the detail of the proposed method in \cref{sec:main}.
In \cref{sec:numer}, we show the experimental results for several ODE models.
In \cref{sec:discussion}, we give concluding remarks.
Several technical details are provided in the appendices.

\section{Problem setting}
\label{sec:setting}

Consider an $M$-dimensional ODE model
\begin{align}
    \frac{\rmd }{\rmd t} x(t; \theta) = f(x(t; \theta),\theta), \quad x(0; \theta) = x_0(\theta), %\quad t\in[0,T],
    \label{eq:ode1}
\end{align}
where $f : \bbR^M \times \bbR^D \to \bbR^M$ and $\theta \in \Theta \subset \bbR^{D}$ is an unknown identifiable parameter.
Suppose that we have noisy observations $y_1,\dots,y_K$ of $J \: (\leq M)$ components of $x(t; \theta)$ at $K$ time points $0\leq t_1 < t_2< \cdots< t_K$:
\begin{align}
    \label{obs_model2} 
    y_k = H x_k(\theta) + \varepsilon_k, \quad \varepsilon_k \sim \rmN (0,\Gamma), \quad k = 1,\dots,K,
\end{align}
where $H \in \bbR^{J \times M}$ is \revise{a given full-rank matrix} (e.g. submatrix of the identity matrix of size $M$), $x_k(\theta) := x(t_k; \theta)$ is the exact solution of \cref{eq:ode1} at $t=t_k$, and $\varepsilon_1,\dots,\varepsilon_K$ are i.i.d. observation noise \revise{with the covariance matrix $\Gamma = \diag (\gamma_1^2,\dots, \gamma_J^2)$.}
We denote the $j$-th element of $y_k$ and $Hx_k$ by $y_{k,j}$ and $H_j x_k$, respectively, for $j=1,\dots,J$.

We consider estimation of $\theta$ based on $y_1,\dots,y_K$.
From \cref{obs_model2}, the likelihood function is given by
\begin{align*}
L(\theta) &= p(y_1,\dots,y_K \mid \theta) \nonumber \\
&= \prod_{k=1}^K \frac{1}{(2\pi)^{J/2} |\Gamma|^{1/2}}\exp \paren*{-\frac12 \paren*{y_k - H x_k(\theta)}^\top \Gamma^{-1} \paren*{y_k - H x_k(\theta)}}.
\end{align*}
Therefore, the maximum likelihood (ML) estimate is the solution of the least squares problem:
\begin{align}\label{ml_form}
\hat{\theta}_{\mathrm{ML}} = \argmax_{\theta \in \Theta} \log L(\theta) =
\argmin_{\theta \in \Theta}
\sum_{k=1}^K
\paren*{y_k - H x_k(\theta)}^\top \Gamma^{-1} \paren*{y_k - H x_k(\theta)}.
\end{align}

Although the maximum likelihood estimate $\hat{\theta}_{\mathrm{ML}}$ in \cref{ml_form} has desirable properties including asymptotic efficiency \cite{ferguson}, it involves the exact solution $x_k(\theta)$ of the ODE model \cref{eq:ode1}, which is not available in practice.
Thus, one conventional method is to substitute a numerical solution $\tx_k(\theta)$ such as Euler and Runge--Kutta \cite{bu16,hn93}, which we call the quasi-maximum likelihood (QML) estimate:
\begin{align} \label{ml_form_approx}
\hat{\theta}_{\mathrm{QML}} = \argmin_{\theta \in \Theta}
\sum_{k=1}^K
\paren*{y_k - H \tx_k(\theta)}^\top \Gamma^{-1} \paren*{y_k - H \tx_k(\theta)}.
\end{align}
If the numerical solution $\tx_k(\theta)$ is sufficiently accurate, then the estimate $\hat{\theta}_{\mathrm{QML}}$ in \cref{ml_form_approx} is considered to be close to the maximum likelihood estimate \cref{ml_form} \revise{from the identifiability assumption}.
However, the discretization error $\tx_k(\theta)-x_k(\theta)$ is not negligible in general, and thus, the estimate $\hat{\theta}_{\mathrm{QML}}$ in \cref{ml_form_approx} has limited estimation accuracy.
We give one simple example in the following.
Note that the asymptotic properties of the estimate $\hat{\theta}_{\mathrm{QML}}$ in \cref{ml_form_approx} was investigated by \cite{xmw10}.

\begin{example}
Consider initial value estimation of the harmonic oscillator:
\begin{align*}
    \frac{\rmd}{\rmd t} x(t;\theta) 
    =
    \begin{bmatrix}
    0 & 1 \\ -1 & 0
    \end{bmatrix}
    x(t;\theta) ,
    \quad
    x(0;\theta) = \theta\in\bbR^2.
\end{align*}
Here, the observation is defined as
\begin{equation*}
    y_k = x(t_k;\theta) + r_k, \quad r_k  \sim {\rm N} (0,\gamma^2I), \quad k = 1,\dots,K,
\end{equation*}
where $t_k = kh$.
The exact solution is given by
\begin{align*}
    x_k (\theta) 
    =
    M^k \theta,
    \quad
    M
    =
    \begin{bmatrix}
    \cos (h) & \sin (h) \\ -\sin (h) & \cos (h)
    \end{bmatrix}.
\end{align*}
On the other hand, the numerical solutions by several ODE solvers are given by $\tx_k (\theta) = \widetilde{M}^k \theta$ with some matrix $\widetilde{M}$. 
For example, the explicit Euler method with step size $\Delta t$ corresponds to
\begin{align*}
    \widetilde{M} = 
    \begin{bmatrix}
    1 & \Delta t \\ -\Delta t & 1
    \end{bmatrix}^{h/\Delta t}.
\end{align*}
Then, the quasi-maximum likelihood estimate $\hat{\theta}_\text{QML}$ in \cref{ml_form_approx} is given by
\begin{align*}
    \hat{\theta}_\text{QML} &= 
    \paren*{\sum_{k=1}^K (\widetilde{M}^k)^\top (\widetilde{M}^k)}^{-1}
    \paren*{\sum_{k=1}^K (\widetilde{M}^k)^\top y_k} \\
    &= \theta + b + \sum_{k=1}^K A_k r_k \sim {\rm N} \left( \theta+b, \gamma^2 \sum_{k=1}^K A_k A_k^{\top} \right),
\end{align*}
where
\begin{align}
    \label{bias_ho}
    b = \paren*{\sum_{k=1}^K (\widetilde{M}^k)^\top (\widetilde{M}^k)}^{-1} \sum_{k=1}^K (\widetilde{M}^k)^\top (M^k-\widetilde{M}^k) \theta,
\end{align}
\begin{align*}
    A_k = \paren*{\sum_{k=1}^K (\widetilde{M}^k)^\top (\widetilde{M}^k)}^{-1} (\widetilde{M}^k)^\top, \quad k=1,\dots,K.
\end{align*}
Therefore, the quasi-maximum likelihood estimate has bias $b \neq 0$ from $\widetilde{M} \neq M$.
The mean squared error (MSE) of the quasi-maximum likelihood estimate is calculated as
\begin{align}
    \label{mse_qml}
    {\rm E}_{\theta} \left[ \| \hat{\theta}_\text{QML}-\theta \|^2 \right] = \| b \|^2 + \gamma^2 \sum_{k=1}^K \mathrm{tr} (A_k A_k^{\top}) = \| b \|^2 + \gamma^2 \mathrm{tr} \paren*{\sum_{k=1}^K (\widetilde{M}^k)^\top (\widetilde{M}^k)}^{-1}.
\end{align}
On the other hand, by putting $\widetilde{M}=M$ in \cref{mse_qml}, the MSE of the maximum likelihood estimate in \cref{ml_form} is obtained as
\begin{align}
    \label{mse_ml}
    {\rm E}_{\theta} \left[ \| \hat{\theta}_\text{ML}-\theta \|^2 \right] = \gamma^2 \mathrm{tr} \paren*{\sum_{k=1}^K ({M}^k)^\top ({M}^k)}^{-1} = \frac{2 \gamma^2}{K}.
\end{align}
\revisee{\cref{fig:HO_MSE} plots the MSE of $\hat{\theta}_\text{QML}$ in \cref{mse_qml} with respect to the data length $K$ for the midpoint rule ($2$nd order) and the Runge--Kutta method ($4$th order), where $\theta=(1,0)^{\top}$, $h=2$ and $\Delta t = 0.5$.}
The MSE of $\hat{\theta}_\text{ML}$ in \cref{mse_ml} is also shown.
It is observed that the MSE of $\hat{\theta}_\text{QML}$ is significantly larger than that of $\hat{\theta}_\text{ML}$.

\color{black}
\input{HO_MSE.tex}
\end{example}

\section{Proposed method}
\label{sec:main}

In this section, we develop a method for parameter estimation in ODE models \cref{eq:ode1} that quantifies the discretization error \revise{based on data.} 
In \cref{subsec:formulation}, we introduce the main idea of modeling the discretization error as random variables and formulate the problem to simultaneous estimation of the discretization error variance  and ODE parameter.
In \cref{subsec:irls}, we propose Iteratively Reweighted Least Squares (IRLS) algorithms where the discretization error variance and ODE parameter are alternately updated.
The update schemes for the discretization error variance and the ODE parameter are explained in \cref{subsec:update_mu,subsec:update_theta}, respectively.

\subsection{Discretization error as random variable}
\label{subsec:formulation}

Let $\xi_k := \tx_k-x_k$ be the discretization error at $k$-th step. 
We model $\xi_1,\dots,\xi_K$ as independent Gaussian random variables:
\begin{align} 
    \label{assmp:dev}
    \xi_k \sim \rmN ( 0, V_k),\quad k=1,\dots,K,
\end{align}
where the \revise{covariance matrix $V_k$} quantifies the magnitude of $\xi_k$.
\revise{Such an idea of modeling the discretization error by random variables has been investigated by recent studies.
We review them in \cref{rem_pn} below.
It may seem unreasonable to assume the discretization error to be independent Gaussian random variables.
However, we introduce this assumption in a pragmatic way to obtain better estimates of ODE parameters through discretization error quantification based on probabilistic models.
We leave the validation and perhaps modification of the modelling assumption to our future work (see \cref{rem_brown} below).}

By substituting \cref{assmp:dev} into \cref{obs_model2}, we obtain
\begin{align}
    \label{model:new}
    y_k = H \tx_k(\theta) + e_k, \quad e_k \sim {\rm N} (0,\Gamma+\Sigma_k),\quad k=1,\dots,K,
\end{align}
where $e_k:=-H\xi_k+\varepsilon_k$ and $\Sigma_k := H V_k H^\top$.
Note that we used the independence of the discretization error $\xi_1,\dots,\xi_K$ and the observation noise $\varepsilon_1,\dots,\varepsilon_K$.
\revise{For simplicity, we assume that each $\Sigma_k$ is diagonal:}
\begin{align*} 
    \Sigma_k = \diag(\sigma_{k,1}^2,\dots,\sigma_{k,J}^2), \quad k=1,\dots,K.
\end{align*} 
\revise{Also, since the discretization error is considered to accumulate in every step of numerical integration, we assume that}
\begin{align} 
    \label{assmp:sigma}
    0 \leq \sigma_{1,j}^2 \leq \sigma_{2,j}^2 \leq\dots\leq \sigma_{K,j}^2 ,\quad j =1,\dots,J.
\end{align}

We estimate the discretization error variance $\Sigma:=( \Sigma_1, \dots,\Sigma_K)$ simultaneously with the ODE parameter $\theta$ by maximum likelihood.
From \cref{model:new}, the likelihood function is given by
\begin{align}
	\label{eq:likelihood}
    L(\theta,\Sigma) 
    &= p(y_1,\dots,y_K \mid \theta,\Sigma) \nonumber \\
    &= \prod_{k=1}^K \frac{1}{(2\pi)^{J/2} |\Gamma+\Sigma_k|^{1/2}}\exp \paren*{-\frac12 \paren*{y_k - H \tx_k(\theta)}^\top (\Gamma+\Sigma_k)^{-1} \paren*{y_k - H \tx_k(\theta)}}.
\end{align}
Thus, the estimate is defined as
\begin{align}
  \label{eq:formulation}
  (\hat{\theta},\hat{\Sigma}) &= \argmax_{(\theta,\Sigma) \in \Theta \times \mathcal{S}} \log L(\theta,\Sigma) \nonumber \\
    &= \argmin_{(\theta,\Sigma) \in \Theta \times \mathcal{S}}
  \sum_{k=1}^K \paren*{
  \log |\Gamma + \Sigma_k|
  +\paren*{y_k - H \tx_k(\theta)}^\top (\Gamma+\Sigma_k)^{-1} \paren*{y_k - H \tx_k(\theta)}} \nonumber \\
  &=\argmin_{(\theta,\Sigma) \in \Theta \times \mathcal{S}}
  \sum_{k=1}^K \sum_{j=1}^J
  \paren*{
  \log (\gamma_j^2 + \sigma_{k,j}^2)
  + \frac{r_{k,j}^2(\theta)}{\gamma_j^2 + \sigma_{k,j}^2}
  },
\end{align}
where $\mathcal{S}$ is the set of $\Sigma$ satisfying the order constraints \cref{assmp:sigma} and $r_{k,j}(\theta)$ is the residual defined by
\begin{align}
  r_{k,j}(\theta) =  y_{k,j} - H_j \tx_k(\theta). \label{res_def}
\end{align}

\revise{Whereas Bayesian probabilistic numerical methods for ODE models have been studied recently, the proposed method is aimed at non-Bayesian parameter estimation.
Therefore, it is free from prior selection and provides a baseline like the maximum likelihood estimator.
Also, the proposed method provides discretization error quantification based on data by solving an inverse problem.}

\begin{remark}\label{rem_pn}
The idea of modeling the discretization error by random variables was discussed recently by, for example, Arnold et al.~\cite{ac13} and Conrad et al.~\cite{cg17}. 
In~\cite{cg17}, by modeling the local truncation error by Gaussian random variables, the authors proposed to quantify the forward uncertainty by doing a number of simulations and applied this method to improve the Bayesian inference.
Convergence analyses are given in~\cite{ls19}.
To preserve properties such as positivity and symplecticity, Abdulle and Garegnani \cite{ag18} proposed to perturb the time step size instead.
\revise{We also note that one referee pointed out the possibility of modeling the discretization error by the uniform distribution supported on the a-priori interval obtained by convergence analysis of the ODE solver.}
On the other hand, several studies have investigated the relationship between ODE solvers and Gaussian process models \cite{sdh14,ssh19,tzc16} and its implications for uncertainty quantification \cite{cc16,kh16,hh14,ssh19,tksh18}.
We also note that probabilistic models have been used in the context of numerical analysis, see, e.g., Hairer~et~al.~\cite{ha08}.
\end{remark}

\begin{remark}\label{rem_brown}
\revisee{
In this study, we model the discretization error by independent Gaussian random variables.
However, this model does not capture the auto-correlation in the discretization error.
Such dependence structures may be flexibly described by the Gaussian process \cite{bo14}.
For example, since the discretization error {accumulates} in every step, it seems more natural to use the Brownian motion (random walk) model:
\begin{align*} 
(e_{1,j}, \dots, e_{K,j})^{\top} \sim {\rm N}_K ( 0,\gamma_j^2 I +\Sigma^{(j)} ), \quad j=1,\dots,J,
\end{align*} 
where
\begin{align*} 
\Sigma^{(j)} = (\sigma_{\min(a,b),j}^2)_{a,b} = \begin{pmatrix} \sigma_{1,j}^2 & \sigma_{1,j}^2 & \sigma_{1,j}^2 & \cdots & \sigma_{1,j}^2 \\ \sigma_{1,j}^2 & \sigma_{2,j}^2 & \sigma_{2,j}^2 & \cdots & \sigma_{2,j}^2 \\ \sigma_{1,j}^2 & \sigma_{2,j}^2 & \sigma_{3,j}^2 & \cdots & \sigma_{3,j}^2 \\ \vdots & \vdots & \vdots & \ddots & \vdots \\ \sigma_{1,j}^2 & \sigma_{2,j}^2 & \sigma_{3,j}^2 & \cdots & \sigma_{K,j}^2 \end{pmatrix}
\end{align*} 
with the order constraint
\begin{align*} 
0 \leq \sigma_{1,j}^2 \leq \sigma_{2,j}^2 \leq\dots\leq \sigma_{K,j}^2.
\end{align*}
Then, the likelihood function is given by
\begin{align*}
&L(\theta,\Sigma) \\
=& p(y_1,\dots,y_K \mid \theta,\Sigma) \nonumber \\
=& \prod_{j=1}^J \frac{1}{(2\pi)^{K/2} |\gamma_j^2 I+\Sigma^{(j)}|^{1/2}}\exp \paren*{-\frac12 (y^{(j)} - H_j \tx(\theta))^\top (\gamma_j^2 I+\Sigma^{(j)})^{-1} (y^{(j)} - H_j \tx(\theta))},
\end{align*}
where $y^{(j)}=(y_{1,j},\dots,y_{K,j})^{\top}$ and $H_j \tx(\theta)=(H_j \tx_1(\theta),\dots,H_j \tx_K(\theta))^{\top}$.
Thus, the estimate is defined as
\begin{align*}
(\hat{\theta},\hat{\Sigma}) &= \argmax_{(\theta,\Sigma) \in \Theta \times \mathcal{S}} \log L(\theta,\Sigma)  \\
&= \argmin_{(\theta,\Sigma) \in \Theta \times \mathcal{S}}
\sum_{j=1}^J \paren*{
	\log |\gamma_j^2 I+\Sigma^{(j)}|
	+(y^{(j)} - H_j \tx(\theta))^\top (\gamma_j^2 I+\Sigma^{(j)})^{-1} (y^{(j)} - H_j \tx(\theta))},
\end{align*}
where $\mathcal{S}$ is the set of $\Sigma$ satisfying the order constraints.
Unfortunately, the above objective function is non-convex with respect to $\Sigma$ and thus the optimization becomes more difficult than \cref{eq:formulation}.
In particular, the efficient procedure of PAVA (see \cref{subsec:update_mu}) is not applicable.
We conducted preliminary experiments with this model by using numerical optimization.
However, it was much more computationally intensive than PAVA and the results were sensitive to the initial guess.
It is an interesting future work to develop an efficient algorithm (e.g. based on Bayesian optimization) for such a more realistic model of the discretization error.
}
\end{remark}

\subsection{Iteratively reweighted least squares}
\label{subsec:irls}
Now, we develop Iteratively Reweighted Least Squares (IRLS) algorithms for solving \cref{eq:formulation}.

We introduce the weights $w=(w_{k,j})$ defined by
\begin{align*}
    w_{k,j} = \frac{1}{\gamma_j^2 + \sigma_{k,j}^2}, \quad 
    k=1,\dots,K, \quad 
    j =1,\dots,J.
\end{align*}
From \cref{assmp:sigma}, the weights have the order constraint:
\begin{align}
    \label{order_w}
    0 < w_{K,j} \leq w_{K-1,j} \leq \cdots \leq w_{1,j} \leq \frac{1}{\gamma_j^2},\quad 
    j =1,\dots,J.
\end{align}
By transforming from $(\theta,\Sigma)$ to $(\theta,w)$, the minimization problem in \cref{eq:formulation} is rewritten as
\begin{align}
    \label{eq:formulation1}
    (\hat{\theta},\hat{w}) = \argmin_{(\theta,w) \in \Theta \times \mathcal{W}} g(\theta,w),
\end{align}
where $\mathcal{W}$ is the set of $w$ satisfying the order constraints \cref{order_w} and the objective function $g$ is defined as
\begin{align}
    \label{g_def}
    g(\theta,w) = \sum_{k=1}^K \sum_{j=1}^J
    \paren*{
    -\log w_{k,j} + w_{k,j} r_{k,j}^2(\theta)
    }.
\end{align}

We solve \cref{eq:formulation1} by alternating minimization with respect to $\theta$ and $w$.
Specifically, starting from an initial guess $\theta^{(0)}$ of the ODE parameter, we iterate the following two steps: 
\begin{align}
  w^{(l)} &=  \argmin_{w \in \mathcal{W}} g(\theta^{(l-1)},w), \label{update_w} \\
  \theta^{(l)} &= \argmin_{\theta \in \Theta} g(\theta,w^{(l)}). \label{update_theta}
\end{align}
The detail of each update will be explained in the following two subsections.
Since the update of $\theta$ in \cref{update_theta} is interpreted as a weighted least squares with weights $w$,
we refer to the algorithm that iterates \cref{update_w} and \cref{update_theta} until convergence as the Iteratively Reweighted Least Squares (IRLS) algorithm (\cref{alg:irls1}). 
As will be shown in the next section, this algorithm often converges in a few iterations.
In addition, we call $L$ iterations of \cref{update_w} and \cref{update_theta} the IRLS($L$) algorithm (\cref{alg:irlsL}) for convenience.
From the order constraint \cref{order_w}, the estimated weights $\hat{w}_{k,j}$ are non-increasing with respect to $k$ and they are interpreted as quantifying the reliability of the numerical solution $\tx_{k,j}(\theta)$ in estimating $\theta$.

\begin{algorithm}[H]
\caption{Iteratively reweighted least squares (IRLS) for solving \cref{eq:formulation1}}
\label{alg:irls1}
\begin{algorithmic}[1]
\STATE{Set the initial guess $\theta^{(0)}$}
\FOR{$l=1,2,\dots$}
\STATE{Solve \cref{update_w} by PAVA (\cref{subsec:update_mu})} 
\STATE{Solve \cref{update_theta} by numerical optimization (\cref{subsec:update_theta})}
\IF{stopping criterion is satisfied}
\STATE{Return $w^{(l)}$ and $\theta^{(l)}$}
\ENDIF
\ENDFOR
\end{algorithmic}
\end{algorithm}

\begin{algorithm}[H]
\caption{Iteratively reweighted least squares $(L)$ (IRLS($L$)) for solving \cref{eq:formulation1}}
\label{alg:irlsL}
\begin{algorithmic}[1]
\STATE{Set the initial guess $\theta^{(0)}$}
\FOR{$l=1,2,\dots,L$}
\STATE{Solve \cref{update_w} by PAVA (\cref{subsec:update_mu})} 
\STATE{Solve \cref{update_theta} by numerical optimization (\cref{subsec:update_theta})}
\ENDFOR
\STATE{Return $w^{(L)}$ and $\theta^{(L)}$}
\end{algorithmic}
\end{algorithm}

\begin{remark} \label{rem_unknown}
\revise{In practice, the observation noise variances $\gamma_1^2,\dots,\gamma_J^2$ may be unknown.
Even in such cases, the proposed algorithms work by using lower bounds $\widetilde{\gamma}_j^2$ on $\gamma_j^2$: $\widetilde{\gamma}_j^2 \leq \gamma_j^2$.
Namely, the order constraint \cref{order_w} is replaced with}
\begin{align*}
    0 < w_{K,j} \leq w_{K-1,j} \leq \cdots \leq w_{1,j} \leq \frac{1}{\widetilde{\gamma}_j^2},\quad 
    j =1,\dots,J.
\end{align*}
\revise{
Such a modification also provides estimates of ${\gamma}_j^2$ as a byproduct.
We validate this modification by simulation in \cref{sec:unknown}.}
\end{remark}

\begin{remark}
\revise{We discuss the convergence property of \cref{alg:irls1,alg:irlsL}.
Since these algorithms are alternating minimization, the objective function monotonically decreases: $g(\theta^{(l+1)},w^{(l+1)}) \leq g(\theta^{(l)},w^{(l)})$.
Also, the function $g(\theta,w)$ in \cref{g_def} is bounded from below.
Therefore, the value $g(\theta^{(l)},w^{(l)})$ of the objective function converges in both algorithms.
However, it does not necessarily imply the convergence of $(\theta^{(l)},w^{(l)})$ due to the non-convexity of the objective function.
Namely, the sequences $(\theta^{(l)},w^{(l)})$ may cycle or stagnate at a non-critical point, like the example in~\cite{po73}.
One way of ensuring the convergence is to add a regularization term to each sub-problem, which is called a proximal point modification~\cite{gs00}.
Our preliminary numerical experiments indicated that such a modification may make the convergence slower.}
\end{remark}

\subsection{Update of the weights}
\label{subsec:update_mu}

Here, we provide an efficient method for the update of the weights $w$ in \cref{update_w}.

The update of the weights \cref{update_w} is solved as follows.
Its proof is given in \cref{app_proof}.

\begin{theorem} \label{thm:pava}
Let $(\mu_{1,j},\dots,\mu_{K,j})$ be the optimal solution of
\begin{align} \label{form_for_pava2}
    \min_{\mu_{1,j}\leq\dots\leq \mu_{K,j}<0} \sum_{k=1}^K \paren*{\Phi (\mu_{k,j})-\mu_{k,j} r_{k,j}(\theta^{(l-1)})^2 },
\end{align}
where $\Phi(\mu)=-\log(-\mu)$, for $j=1,\dots,J$.
Then, the solution of \cref{update_w} is given by $w^{(l)}_{k,j}=-\max(\mu_{k,j},-1/\gamma_j^2)$ for $j=1,\dots,J$.
\end{theorem}

The optimization problem \cref{form_for_pava2} is efficiently solved by an algorithm called the pool adjacent violators algorithm (PAVA)~\cite{bb72,rwd88,vE06}.
Specifically, let $S_0 = 0$ and $S_k = r_{1,j}^2+\dots + r_{k,j}^2$ for $k = 1,\cdots,K$. 
Consider the line graph connecting $(0,S_0),(1,S_1),\dots,(K,S_K)$ (dashed line in \cref{fig:gcm}).
We can efficiently compute the maximal convex function, called the greatest convex minorant, on $[0,K]$ which lies entirely below this graph (solid line in \cref{fig:gcm}).
Let $\overline{S}_k$ be the value of the greatest convex minorant at $k$ for $k=0,1,\dots,K$.
Then, the optimal solution to \cref{form_for_pava2} is given by
\begin{align*} 
    \hat{\mu}_k= \phi^{-1} (\overline{S}_k - \overline{S}_{k-1}) = -\frac{1}{\overline{S}_k - \overline{S}_{k-1}}, \quad k = 1,\dots, K,
\end{align*}
where $\phi(\mu)=\Phi'(\mu)=-1/\mu$.
See \cite{rwd88} for the proof.

The computational cost for updating the weights is much cheaper than that for a single numerical integration of the system of ODEs \cref{eq:ode1} in most cases.
Note also that the weights $w=(w_{k,j})$ for different $j$ can be updated in parallel.

\begin{figure}[htbp]
    \centering
    \begin{tikzpicture}[domain=0:4]
    \draw[->] (-0.2,0) -- (5.2,0) node[right] {$k$};
    \draw[->] (0,-0.2) -- (0,4.2) node[above] {$S$};
    \draw[thick,dashed] (0,0) -- (1,0.8) -- (2,0.9) --(3,1.6) -- (4,3.6) -- (5,4);
    \draw[thick, color=black, opacity=0.7] (0,0) -- (2,0.9) --(3,1.6) --(5,4);
    \draw[] (1,0) -- (1,-0.1) node[below] {$1$};
    \draw[] (2,0) -- (2,-0.1) node[below] {$2$};
    \draw[] (3,0) -- (3,-0.1) node[below] {$3$};
    \draw[] (4,0) -- (4,-0.1) node[below] {$4$};
    \draw[] (5,0) -- (5,-0.1) node[below] {$5$};
    \node[left] at (0,0.3) {\small $(0,S_0)$};
    \node[above] at (0.7,0.8) {\small $(1,S_1)$};
    \node[below] at (2.2,0.8) {\small $(2,S_2)$};
    \node[right] at (3,1.6) {\small $(3,S_3)$};
    \node[left] at (4,3.6) {\small $(4,S_4)$};
    \node[below] at (5.4,4) {\small $(5,S_5)$};
    \end{tikzpicture}
    
    \caption{The line graph with $K=5$ (dashed line) and its greatest convex minorant (solid line).}
    \label{fig:gcm}
\end{figure}
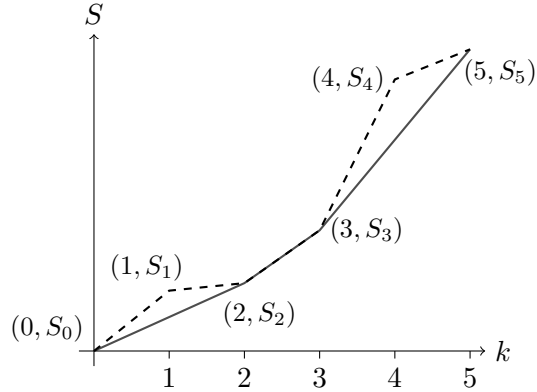

\begin{remark}
PAVA is a general algorithm for the problem of estimating ordered natural parameters of exponential families, which is called the generalized isotonic regression \cite{bb72,rwd88,vE06}, and the update of the weights is naturally interpreted in this context.
Namely, from \cref{model:new}, {the} square of the residual \cref{res_def} follows the chi-square distribution:
\begin{align*}
    r_{k,j}(\theta)^2 \sim (\gamma_j^2+\sigma_{k,j}^2) \chi_1^2.
\end{align*}
In other words, {$r_{k,j}(\theta)^2$} follows the Gamma distribution with shape parameter $1/2$ and scale parameter $w_{k,j}^{-1}=\gamma_j^2+\sigma_{k,j}^2$:
\begin{align*}
    r_{k,j}(\theta)^2 \sim {\rm Gamma} \left( \frac{1}{2}, \frac{2}{w_{k,j}} \right).
\end{align*}
In this way, the procedure above can be interpreted as applying PAVA to the maximum likelihood estimation of ordered scale parameters of Gamma distributions (see Example 1.5.3 in \cite{rwd88}).
\end{remark}

\begin{remark}\label{rem:conrad}
\revisee{
Instead of solving \cref{update_w}, we can employ the probabilistic ODE solver by \cite{cg17} for updating the weights.
Specifically, by solving the forward problem many times with Gaussian perturbation in each step, the discretization error variance $\sigma_{k,j}^2$ can be estimated.
Whereas this method enables discretization error quantification without the monotonicity condition \cref{assmp:sigma} (see \cref{fig:Lorenz_variables_error_conrad}), it takes much more computational cost than PAVA, which only requires one integration.
We confirmed that PAVA attains comparable estimation accuracy to the probabilistic ODE solver in simulation (\cref{fig:Lorenz_comparison_conrad}).
}
\end{remark}

\subsection{Update of the ODE parameter}
\label{subsec:update_theta}

The update of the ODE parameter $\theta$ in \cref{update_theta} is interpreted as solving the weighted least squares problem:
\begin{align}
    \label{update_theta_opt}
    \theta^{(l)} = \argmin_{\theta \in \Theta} \tilde{R}(\theta),
\end{align}
where
\begin{align}
    \tilde{R}(\theta) = \sum_{k=1}^K \tilde{R}_k (\theta), \quad \tilde{R}_k (\theta) = \sum_{j=1}^J w^{(l)}_{k,j} (y_{k,j}-H_j \tx_k(\theta))^2.
    \label{obfun}
\end{align}
We employ numerical optimization to solve \cref{update_theta_opt}.
In order to use a gradient method such as the quasi-Newton method and nonlinear conjugate gradient method, the gradient $\nabla_{\theta} \tilde{R}(\theta)$ of the objective function $\tilde{R}(\theta)$ is necessary.
However, since the objective function $\tilde{R}(\theta)$ is implicitly defined through the ODE solver $\tx_k(\theta)$, computation of its gradient is not trivial.
In this subsection, we briefly review a method for efficiently computing the exact gradient of $\tilde{R}(\theta)$. 
The key ingredients are the adjoint system and symplectic partitioned Runge--Kutta scheme~\cite{sa92}. 
See \cite{ss16} for more detail.

\revise{First,} let us focus on an $M$-dimensional ODE model with unknown initial state:
\begin{align}
    \label{eq:original1}
    \frac{\rmd}{\rmd t} x(t; \theta) = {f}(x(t; \theta)), \quad x(0; \theta) = \theta.
\end{align}
\revise{More general case with unknown system parameter will be treated in the last paragraph of this subsection.}
The adjoint system of \cref{eq:original1} is defined as
\begin{align} \label{eq:adjoint1}
    \frac{\rmd}{\rmd t} \lambda (t) = - \nabla_x {f}(x(t; \theta))^\top \lambda(t),
\end{align}
where $\lambda(t) \in \bbR^{M}$ and $\nabla_x {f}(x(t; \theta))$ denotes the Jacobian matrix of ${f}$ at \revise{$
x(t; \theta)$}.
For a function of the form $R(\theta)=J(x(T;\theta))$, the backward solution of the adjoint system \cref{eq:adjoint1} with the final state condition $\lambda (T) = \nabla_x J (x(T;\theta))$ satisfies $\lambda (0) = \nabla_\theta R (\theta)$ (see Proposition 3.1 in \cite{ss16}).
Similar relation holds when the original system \cref{eq:original1} and adjoint system \cref{eq:adjoint1} are discretized by a symplectic partitioned Runge--Kutta scheme~\cite{sa92}, which is a well-known scheme in the context of \emph{Geometric Integration}~\cite{hl06,sa94} (see \cref{app_gi} for a brief introduction of geometric integration and symplectic partitioned Runge--Kutta methods). 
Specifically, let $\tilde{x}_k(\theta)$ and $\tilde{\lambda}_k$ be the numerical solution of \cref{eq:original1} and \cref{eq:adjoint1} by a symplectic partitioned Runge--Kutta scheme.
Then, for a function of the form $\tilde{R}(\theta)=J(\tilde{x}_K(\theta))$, the backward solution of the adjoint system \cref{eq:adjoint1} with the final state condition $\tilde{\lambda}_K = \nabla_x J (\tilde{x}_K(\theta))$ satisfies $\tilde{\lambda}_0 = \nabla_\theta \tilde{R} (\theta)$ (see Theorem 3.4 in \cite{ss16}).
For example, if the ODE model \cref{eq:original1} is discretized by the explicit Euler scheme
\begin{align}
    \label{eq:exEuler}
    \frac{\tilde{x}_{k+1}(\theta)-\tilde{x}_k(\theta)}{t_{k+1}-t_k} = f(\tilde{x}_k(\theta)),
\end{align}
then the adjoint system \cref{eq:adjoint1} should be discretized as
\begin{align}
    \label{eq:Euler_sym}
    \frac{\tlambda_{k+1}-\tlambda_k}{t_{k+1}-t_k} =-\nabla_x f(\tilde{x}_{k}(\theta)) ^\top \tlambda_{k+1},
\end{align}
which is an explicit scheme for backward integration. 

\revise{Now, since the adjoint system \cref{eq:adjoint1} is linear, the gradient of the objective function $\tilde{R}(\theta)$ in \cref{update_theta_opt}, which is the sum of $\tilde{R}_k(\theta)$ in \cref{obfun}, is computed by \cref{alg:gradient} when the explicit Euler method~\cref{eq:exEuler} is employed for the numerical solution.}
Namely, the output $\tlambda_0$ of \cref{alg:gradient} coincides with $\nabla_\theta \tilde{R}(\theta)$.
\revise{The calculation is exact up to round-off.}
This procedure is extended straightforwardly to more general Runge--Kutta schemes.
In \cref{subsec:Kepler}, we will use the Störmer--Verlet method (see \cref{app_sv}) for the numerical solution of the Kepler equation.
Also note that, in practice, we can adopt a composition of Runge--Kutta schemes with time step sizes smaller than the observation interval $t_{k+1}-t_k$, as will be done in \cref{sec:numer}.

\begin{algorithm}[H]
    \caption{Exact calculation of the gradient $\nabla_\theta \tilde{R}(\theta)$ of $\tilde{R}(\theta)$ in \cref{obfun}}
\label{alg:gradient}
\begin{algorithmic}[1]
\STATE{Compute the numerical solutions $\tx_1,\dots,\tx_K$ for \cref{eq:original1} using the Euler method \cref{eq:exEuler}}
\STATE{Set $\tlambda_K = \nabla_x \tilde{R}_k (\tx_K (\theta))$}
\FOR{$k=1,\dots,K$}
\STATE{Calculate $\tlambda_{K-k}$ using \cref{eq:Euler_sym}}
\IF{$k=K$}
\STATE{Return $\tlambda_0$}
\ENDIF
\STATE{$\tlambda_{K-k} \leftarrow \tlambda_{K-k} + \nabla_x \tilde{R}_{K-k} (\tx_{K-k} (\theta))$ }
\ENDFOR
\end{algorithmic}
\end{algorithm}

{Although we focused on the ODE model with unknown initial state \cref{eq:original1}, the above method is applicable to the ODE model \cref{eq:ode1} with unknown system parameter as well.
Specifically, if the ODE model \cref{eq:ode1} includes the unknown system parameter $\theta_S$, then it is reduced to the form of \cref{eq:original1} as follows:
\begin{align*}
    \frac{\rmd}{\rmd t} z(t; \theta) = \begin{bmatrix} f(z(t; \theta)) \\ 0 \end{bmatrix}, \quad z(0; \theta) = \begin{bmatrix} x_0(\theta) \\ \theta_S \end{bmatrix},
\end{align*}
where
\begin{align*}
    z(t; \theta) = \begin{bmatrix} x(t; \theta) \\ u(t; \theta) \end{bmatrix}.
\end{align*}
\revise{Here, we introduced the constant variable $u(t;\theta) \equiv \theta_S$ so that the unknown system parameter can be treated in the same way with the unknown initial state.}
Thus, \cref{alg:gradient} efficiently computes the gradient of $\tilde{R}(\theta)$ in \cref{obfun}.
}
\revisee{
\begin{remark}
Recently, the method of \cite{ss16} has been extended to efficiently compute the exact Hessian-vector multiplication $\nabla_{\theta}^2 \tilde{R}(\theta) \cdot v$ by using the second-order adjoint system \cite{im19}.
This method enables to use a second-order optimization algorithm such as the Newton method to minimize $\tilde{R}(\theta)$.
\end{remark}
}
\subsection{Confidence intervals}
\label{subsec:confidence}
\revisee{
To obtain uncertainty quantification of the point estimate $\hat{\theta}$ by the IRLS algorithm, we develop a method for constructing confidence intervals of ${\theta}$.

For the likelihood function $L(\theta,\Sigma)$ in \cref{eq:likelihood}, let
\begin{align*}
	l_i (\theta_i) = \max_{\theta^{(-i)},\Sigma} \ \log L (\theta,\Sigma), \quad i=1,\dots,D,
\end{align*}
be the profile likelihood for $\theta_i$, where $\theta^{(-i)}=(\theta_1,\dots,\theta_{i-1},\theta_{i+1},\dots,\theta_D)$.
Note that the IRLS estimate $\hat{\theta}$ satisfies
\begin{align*}
	l_i(\hat{\theta}_i)=\max_{\theta_i} \ l_i(\theta_i), \quad i=1,\dots,D.
\end{align*}
Under the null hypothesis $\theta_i=\theta_{i0}$, the profile likelihood-ratio statistic
\begin{align*}
	\mathrm{LR}_i = 2 (l_i (\hat{\theta}_i) - l_i(\theta_{i0}))
\end{align*}
is asymptotically distributed with the chi-square distribution with one degree of freedom.
Thus, we construct a 95 \% confidence interval $[\underline{\theta}_i, \bar{\theta}_i]$ of $\theta_i$ by using $\underline{\theta}_i$ and $\bar{\theta}_i$ defined by
\begin{align*}
	l_i(\underline{\theta}_i) = l_i(\bar{\theta}_i) = l_i(\hat{\theta}_i)- \frac{1}{2}(\chi^2_1)^{-1} (0.95), \quad \underline{\theta}_i \leq \hat{\theta}_i \leq \bar{\theta}_i.
\end{align*}
If $(\hat{\theta},\hat{\Sigma})$ lies in the interior of the parameter space, then $\underline{\theta}_i$ and $\bar{\theta}_i$ can be found by using the Newton--Raphson algorithm \cite{vm88}.
However, it is not applicable here because some of the order constraints on $\Sigma$ are usually active.
Thus, we use the bisection method to find $\underline{\theta}_i$ and $\bar{\theta}_i$.
}

\section{Numerical experiments}
\label{sec:numer}
In this section, we investigate the performance of the IRLS algorithms by numerical experiments on three ODE models: Lorenz system, FitzHugh--Nagumo model and Kepler equation. 
Experiments were performed in a computation environment with 1.6 GHz Intel Core i5, 16 GB memory, Mac OS X 10.14.4, and MATLAB (R2019a).

In each experiment, we generated observation data $y_1,\dots,y_K$ by solving the ODE model \cref{eq:ode1} with the MATLAB function \textsf{ode45} and then adding Gaussian observation noise following \cref{obs_model2}.
In \textsf{ode45}, both relative and absolute error tolerances were set to $3.0\times 10^{-14}$, which is much smaller than the observation noise variance employed below.

For the Lorenz system, we used the explicit Euler method\revise{, the Heun method and the (4-stage 4-th order) Runge--Kutta (RK) method to obtain} the numerical solution $\tx(t; \theta)$.
For the FitzHugh--Nagumo model, we employed the explicit Euler method.
For the Kepler equation, we used the Störmer--Verlet method, which is a symplectic integrator for Hamiltonian systems~\cite{hl06} (see \cref{app_sv}).
Note that the time step size was set to be smaller than the observation interval in all cases.

To solve the weighted least squares in \cref{update_theta}, we used the MATLAB function \textsf{fminunc} (quasi-Newton algorithm; BFGS) with the gradient computed by the method in \cref{subsec:update_theta}.
We have also tested \textsf{fminunc} (trust-region algorithm) and the nonlinear conjugate gradient method, but the quasi-Newton algorithm was the fastest in many cases.

As shown in \cref{fig:Lorenz_estimation1}, the IRLS algorithms converged within a few iterations in most cases.
Thus, instead of providing a specific stopping criterion for IRLS (\cref{alg:irls1}), we simply identify IRLS(20) (\cref{alg:irlsL}) as IRLS (\cref{alg:irls1}).

\subsection{Lorenz system}
Here, we consider the Lorenz system~\cite{lo63}:
\begin{align*}
\frac{\rmd}{\rmd t} \begin{bmatrix}
x_1 \\ x_2 \\ x_3
\end{bmatrix}
=\begin{bmatrix}
\sigma (-x_1 + x_2) \\
x_1 (\rho-x_3)-x_2 \\
x_1x_2-\beta x_3
\end{bmatrix}, \quad
\begin{bmatrix}
x_1(0) \\ x_2(0) \\ x_3(0)
\end{bmatrix}
=
\begin{bmatrix}
-10 \\ -1 \\ 40
\end{bmatrix},
\end{align*}
where $(\sigma , \rho, \beta ) = (10, 28, 8/3)$.

We {consider} simultaneous estimation of {both the initial state and system parameter $\theta = (x_1(0),x_2(0),$ $x_3(0),\sigma , \rho, \beta )$}.
\revise{In \cref{subsubsec:Lorenz_performance}, \cref{subsubsec:Lorenz_uq} and \cref{sec:unknown}, the observations of $(x_1,x_2,x_3)$ with observation noise variance $\Gamma = \diag (0.5,0.1,0.1)$} are taken at $t_k=(k-1)h$ for $k=1,\dots,K$, where $h=0.01$ and $K=201$.  Note that similar results were obtained for other settings of $\Gamma$.
\revise{In \cref{subsubsec:Lorenz_limited}, we consider the case of limited observation: only $(x_1,x_2)$ is observed with observation noise variance $\Gamma = \diag (0.5,0.1)$ at $t_k=(k-1)h$ for $k=11,\dots,K$, where $h=0.01$ and $K=201$.}
The initial guess is set to $\theta^{(0)} = [-9, -1.5, 39, 11,29,3]$.
For the ODE solver $\tx_k(\theta)$, we employ \revise{the explicit Euler, Heun and RK methods.} 

\revise{To evaluate the estimation accuracy, we use the error measure defined by}
\begin{align}
    \label{Lorenz_err}
    l(\hat{\theta},\theta) = \int_a^b \| x(t; \hat{\theta})-x(t; {\theta}) \|^2 {\rm d} t,
\end{align}
\revise{where $[a,b]$ is the observation interval.
This error measure naturally accounts for the complexity of the map from the parameter $\theta$ to the trajectory $x(t;\theta)$.
To compute the integral in \cref{Lorenz_err}, we employ the trapezoidal rule with $10^4$ points and regard the outputs of \textsf{ode45} as $x(t; \theta)$ and $x(t; \hat{\theta})$.}

\subsubsection{Performance of IRLS algorithms}
\label{subsubsec:Lorenz_performance}
First, we check the performance of the IRLS algorithms. 

\cref{fig:Lorenz_estimation1} plots the objective function $g(\theta^{(l)},w^{(l)})$ in \cref{g_def} and error $l({\theta}^{(l)},\theta)$ in \cref{Lorenz_err}
with respect to the iteration count $l$ for each \revise{ODE solver}. 
The initial guess and step size were set to $\theta^{(0)} = [-9, -1.5, 39, 11,29,3]$ and $\Delta t = 5.0 \times 10^{-3}$, respectively. 
The IRLS algorithms converge within a few iterations, and the estimation accuracy is better for {higher order solvers.}

% fss = 1, 20, 100
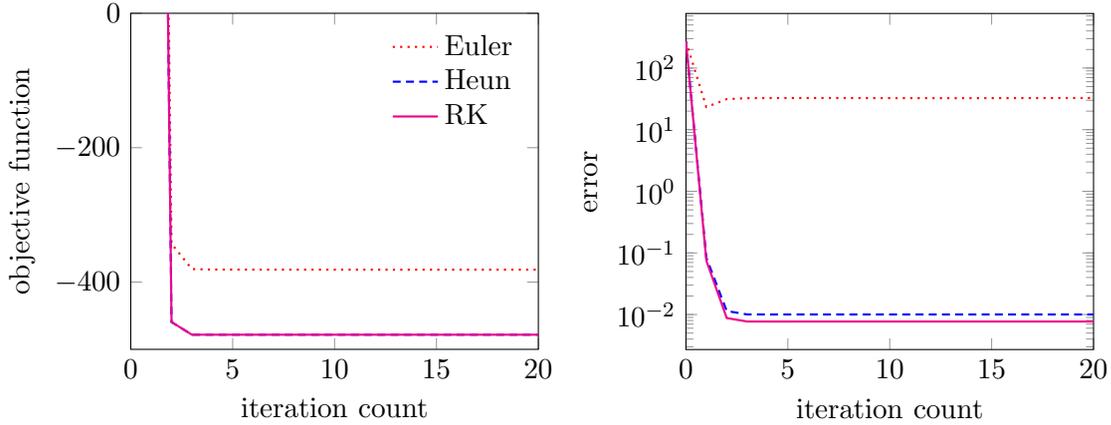
\begin{figure}[htbp]
\centering
\begin{tabular}{ll}
\begin{tikzpicture}
\tikzstyle{every node}=[]
\begin{axis}[width=7cm,%restrict x to domain = 0:69,
xmax=20,xmin=0,
%ymax=100,
%%	restrict y to domain=-0.8:1.1, ymax=1.1,ymin=-0.8,
ymax = 0,ymin = -500,
xlabel={iteration count},ylabel={objective function},
ylabel near ticks,
legend entries={Euler,Heun,RK},
legend style={legend cell align=left,draw=none,fill=none,legend pos=north east},
	]
%\node[above right] at (axis cs: -2,-1) {$h=0.1$};
%\node[above left] at (axis cs: 0.5,-1.05) {$60$ steps};
\addplot[thick, dotted, color=red,
filter discard warning=false, unbounded coords=discard
] table {
0 3471.38210697
1 2165.29299966
2 -342.90640973
3 -380.58818240
4 -381.29740520
5 -381.38789562
6 -381.43746166
7 -381.46915813
8 -381.49014094
9 -381.50438681
10 -381.51426238
11 -381.52122896
12 -381.52621661
13 -381.52983275
14 -381.53248301
15 -381.53444355
16 -381.53590562
17 -381.53700361
18 -381.53783323
19 -381.53846341
20 -381.53894432
};
\addplot[thick, densely dashed, color=blue,
filter discard warning=false, unbounded coords=discard
] table {
0 3467.37891293
1 2154.97814840
2 -459.40803521
3 -478.25906789
4 -478.26546282
5 -478.26546660
6 -478.26546661
7 -478.26546661
8 -478.26546661
9 -478.26546661
10 -478.26546661
11 -478.26546661
12 -478.26546661
13 -478.26546661
14 -478.26546661
15 -478.26546661
16 -478.26546661
17 -478.26546661
18 -478.26546661
19 -478.26546661
20 -478.26546661
};
\addplot[thick, color=magenta,
filter discard warning=false, unbounded coords=discard
] table {
0 3467.39002395
1 2155.02272961
2 -459.34895644
3 -478.23595756
4 -478.24235741
5 -478.24236119
6 -478.24236120
7 -478.24236120
8 -478.24236120
9 -478.24236120
10 -478.24236120
11 -478.24236120
12 -478.24236120
13 -478.24236120
14 -478.24236120
15 -478.24236120
16 -478.24236120
17 -478.24236120
18 -478.24236120
19 -478.24236120
20 -478.24236120
};
\end{axis}
\end{tikzpicture} 
& \hspace{-1em}
\begin{tikzpicture}
\tikzstyle{every node}=[]
\begin{axis}[width=7cm,%restrict x to domain = 0:69,
xmax=20,xmin=0,
%ymax=100,
%%	restrict y to domain=-0.8:1.1, ymax=1.1,ymin=-0.8,
ymode = log,
xlabel={iteration count},ylabel={error},
ylabel near ticks,
	]
%\node[above right] at (axis cs: -2,-1) {$h=0.1$};
%\node[above left] at (axis cs: 0.5,-1.05) {$60$ steps};
\addplot[thick, dotted, color=red,
filter discard warning=false, unbounded coords=discard
] table {
0 269.67129439
1  23.29193660
2  31.29627373
3  32.44825983
4  32.55472007
5  32.60161347
6  32.61853031
7  32.61188539
8  32.57127270
9  32.53039622
10  32.49462375
11  32.46355529
12  32.43097641
13  32.42210722
14  32.45016662
15  32.48264889
16  32.51889522
17  32.55606212
18  32.58982145
19  32.62343057
20  32.66566177
};
\addplot[thick, densely dashed, color=blue,
filter discard warning=false, unbounded coords=discard
] table {
0 269.67129439
1   0.08375800
2   0.01141503
3   0.01004958
4   0.01003272
5   0.01003316
6   0.01003319
7   0.01003319
8   0.01003319
9   0.01003319
10   0.01003319
11   0.01003319
12   0.01003319
13   0.01003319
14   0.01003319
15   0.01003319
16   0.01003319
17   0.01003319
18   0.01003319
19   0.01003319
20   0.01003319
};
\addplot[thick, color=magenta,
filter discard warning=false, unbounded coords=discard
] table {
0 269.67129439
1   0.07487265
2   0.00878359
3   0.00771364
4   0.00770420
5   0.00770473
6   0.00770477
7   0.00770477
8   0.00770477
9   0.00770477
10   0.00770477
11   0.00770477
12   0.00770477
13   0.00770477
14   0.00770477
15   0.00770477
16   0.00770477
17   0.00770477
18   0.00770477
19   0.00770477
20   0.00770477
};
\end{axis}
\end{tikzpicture} 
\end{tabular}

\caption{
Objective function $g(\theta^{(l)},w^{(l)})$ in \eqref{g_def} and error $l(\theta^{(l)},\theta)$ in \eqref{Lorenz_err} of IRLS (Algorithm~\ref{alg:irls1}) for the Lorenz system.
%The results with the Heun and RK methods are almost identical.
}
\label{fig:Lorenz_estimation1}
\end{figure}

\cref{fig:Lorenz_comparison1} plots the error in \cref{Lorenz_err} of the IRLS algorithms and the conventional method in \cref{ml_form_approx} with respect to the step size $\Delta t$. 
\revise{
For the Euler method with large step size ($\Delta t \geq 5.0 \times 10^{-3}$), only a single iteration (IRLS$(1)$) gives better estimates than the conventional method.
For the Euler method with middle step size ($5.0 \times 10^{-3} \leq \Delta t \leq 5.0 \times 10^{-5}$), the estimation accuracy of the IRLS algorithms is almost the same with the conventional method. 
For the Euler method with small step size ($\Delta t \leq 5.0 \times 10^{-5}$), the IRLS algorithms with more than one iterations provide better estimates than the conventional method.
On the other hand, for the RK method, the estimation accuracy is almost independent from the step size and the IRLS algorithms with more than one iterations consistently dominate the conventional method.
The results for the Heun method were almost the same with the RK method.
}

% fss = 1, 20, 100
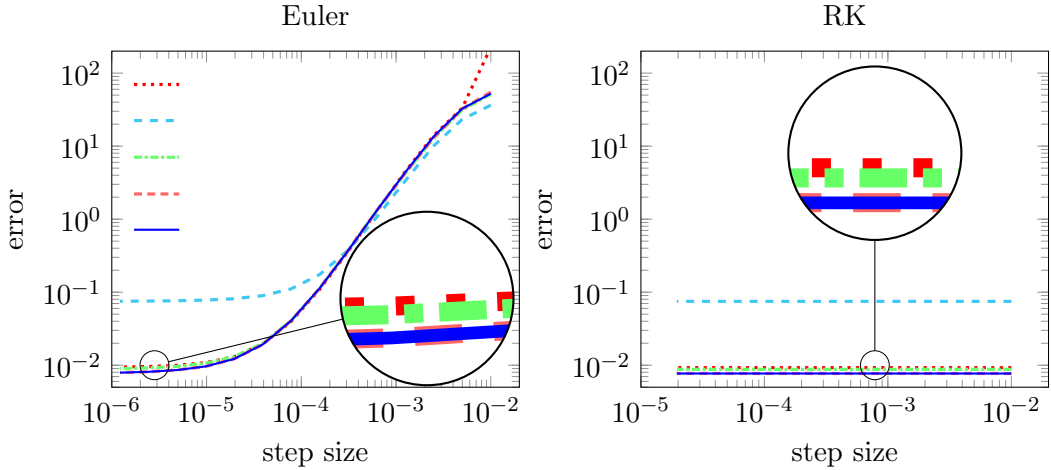
\begin{figure}[htbp]
\centering
\begin{tikzpicture}[spy using outlines=
  {circle, magnification=6, connect spies}]
%\tikzstyle{every node}=[font=\footnotesize]
\tikzstyle{every node}=[]
\begin{axis}[width=7cm,%restrict x to domain = 0:69,
xmax=0.02,xmin=0.000001,
ymax=200, ymin = 0.005,
%%	restrict y to domain=-0.8:1.1, ymax=1.1,ymin=-0.8,
ymode = log,
xmode = log,
xlabel={step size},ylabel={error},
title = {Euler},
ylabel near ticks,
legend entries={conventional,IRLS(1),IRLS(2),IRLS(3),IRLS},
legend style={legend cell align=left,draw=none,fill=white,fill opacity=0,text opacity=1,legend pos=north west,},
	]
%\node[above right] at (axis cs: 0.00005,8) {{\color{red}standard least squares}};
%\node[above right] at (axis cs: 0.0003,1) {{\color{blue}IRLS}};
%\node[above right] at (axis cs: 0.0004,0.03) {{\color{cyan!60} \bf IRLS (1 iteration)}};
\addplot[very thick, color=red, dotted, 
filter discard warning=false, unbounded coords=discard
] table {
0.01 226.3 
0.005 33.41 
0.0025 14.01 
0.00125 4.491 
0.000625 1.285 
0.0003125 0.3625 
0.0001563 0.1095 
7.813e-05 0.03963 
3.906e-05 0.01943 
1.953e-05 0.01309 
9.766e-06 0.01086 
4.883e-06 0.009983 
2.441e-06 0.009605 
1.221e-06 0.009431 
};
\addplot[very thick, dashed, color=cyan!60,
filter discard warning=false, unbounded coords=discard
] table {
0.01 36.41 
0.005 23.29 
0.0025 10.12 
0.00125 3.354 
0.000625 1.064 
0.0003125 0.3774 
0.0001563 0.1745 
7.813e-05 0.1112 
3.906e-05 0.08964 
1.953e-05 0.08141 
9.766e-06 0.07793 
4.883e-06 0.07635 
2.441e-06 0.0756 
1.221e-06 0.07523 
};
\addplot[very thick, densely dashdotted, color=green!60,
filter discard warning=false, unbounded coords=discard
] table {
0.01 51.51 
0.005 31.3 
0.0025 12.87 
0.00125 4.237 
0.000625 1.306 
0.0003125 0.377 
0.0001563 0.113 
7.813e-05 0.04084 
3.906e-05 0.01985 
1.953e-05 0.01308 
9.766e-06 0.01063 
4.883e-06 0.009628 
2.441e-06 0.009187 
1.221e-06 0.00898 
};
\addplot[very thick, densely dashed, color=red!60,
filter discard warning=false, unbounded coords=discard
] table {
0.01 55.06 
0.005 32.45 
0.0025 13.15 
0.00125 4.281 
0.000625 1.307 
0.0003125 0.3732 
0.0001563 0.1132 
7.813e-05 0.04067 
3.906e-05 0.01918 
1.953e-05 0.01218 
9.766e-06 0.009633 
4.883e-06 0.008594 
2.441e-06 0.008134 
1.221e-06 0.007919 
};
\addplot[thick, color=blue,
filter discard warning=false, unbounded coords=discard
] table {
0.01 52.28 
0.005 32.67 
0.0025 13.6 
0.00125 4.38 
0.000625 1.307 
0.0003125 0.3725 
0.0001563 0.1133 
7.813e-05 0.04073 
3.906e-05 0.0192 
1.953e-05 0.01219 
9.766e-06 0.00963 
4.883e-06 0.008588 
2.441e-06 0.008127 
1.221e-06 0.007911 
};
\coordinate (spypoint) at (axis cs: 0.000002812,0.01);
\coordinate (magnifyglass) at (axis cs: 0.002112, 0.082);
\end{axis}
\spy [black, size=2.3cm] on (spypoint)
  in node[fill=white] at (magnifyglass);
\end{tikzpicture} 
\begin{tikzpicture}[spy using outlines=
  {circle, magnification=6, connect spies}]
%\tikzstyle{every node}=[font=\footnotesize]
\tikzstyle{every node}=[]
\begin{axis}[width=7cm,%restrict x to domain = 0:69,
xmax=0.02,xmin=0.00001,
ymax=200, ymin = 0.005,
%%	restrict y to domain=-0.8:1.1, ymax=1.1,ymin=-0.8,
ymode = log,
xmode = log,
xlabel={step size},ylabel={error},
title = {RK},
ylabel near ticks,
	]
%\node[above right] at (axis cs: 0.00005,8) {{\color{red}standard least squares}};
%\node[above right] at (axis cs: 0.0003,1) {{\color{blue}IRLS}};
%\node[above right] at (axis cs: 0.0004,0.03) {{\color{cyan!60} \bf IRLS (1 iteration)}};
\addplot[very thick, color=red, dotted, 
filter discard warning=false, unbounded coords=discard
] table {
0.01 0.009255 
0.005 0.009266 
0.0025 0.009266 
0.00125 0.009266 
0.000625 0.009266 
0.0003125 0.009266 
0.0001563 0.009266 
7.813e-05 0.009266 
3.906e-05 0.009266 
1.953e-05 0.009266 
};
\addplot[very thick, dashed, color=cyan!60,
filter discard warning=false, unbounded coords=discard
] table {
0.01 0.07484 
0.005 0.07487 
0.0025 0.07487 
0.00125 0.07487 
0.000625 0.07487 
0.0003125 0.07487 
0.0001563 0.07487 
7.813e-05 0.07487 
3.906e-05 0.07487 
1.953e-05 0.07487 
};
\addplot[very thick, densely dashdotted, color=green!60,
filter discard warning=false, unbounded coords=discard
] table {
0.01 0.008774 
0.005 0.008784 
0.0025 0.008784 
0.00125 0.008784 
0.000625 0.008784 
0.0003125 0.008784 
0.0001563 0.008784 
7.813e-05 0.008784 
3.906e-05 0.008784 
1.953e-05 0.008784 
};
\addplot[very thick, densely dashed, color=red!60,
filter discard warning=false, unbounded coords=discard
] table {
0.01 0.007704 
0.005 0.007714 
0.0025 0.007714 
0.00125 0.007714 
0.000625 0.007714 
0.0003125 0.007714 
0.0001563 0.007714 
7.813e-05 0.007714 
3.906e-05 0.007714 
1.953e-05 0.007714 
};
\addplot[thick, color=blue,
filter discard warning=false, unbounded coords=discard
] table {
0.01 0.007695 
0.005 0.007705 
0.0025 0.007705 
0.00125 0.007705 
0.000625 0.007705 
0.0003125 0.007705 
0.0001563 0.007705 
7.813e-05 0.007705 
3.906e-05 0.007705 
1.953e-05 0.007705 
};
\coordinate (spypoint) at (axis cs: 0.0007812,0.01);
\coordinate (magnifyglass) at (axis cs: 0.0007812, 8);
\end{axis}
\spy [black, size=2.3cm] on (spypoint)
  in node[fill=white] at (magnifyglass);
\end{tikzpicture} 

\caption{Errors $l(\hat{\theta},\theta)$ in \eqref{Lorenz_err} of IRLS (Algorithm~\ref{alg:irls1}), IRLS($L$) (Algorithm~\ref{alg:irlsL}) with several values of $L$, and the conventional method in \eqref{ml_form_approx} for the Lorenz system.  
%\revise{The Euler method is compared with the Heun methods.}
%\sout{Left: $\theta^{(0)} = [-9, -1.5, 39]^{\top}$. Right: $\theta^{(0)} =[-9.9, -0.9, 40.5]^{\top}$.}
}
\label{fig:Lorenz_comparison1}
\end{figure}

\cref{fig:Lorenz_estimation2} plots the estimated weights $w_{k,j}^{(20)}$ at the 20th iteration for the Euler methods with $\Delta t = 5.0 \times 10^{-3}$ and $\Delta t = 1.0 \times 10^{-3}$ and RK method with $\Delta t = 5.0 \times 10^{-3}$.
The result for the Heun method was almost the same with the RK method.
From \cref{order_w}, the weight $w_{k,j}$ does not exceed $1/\gamma_j^2$, which is \revise{2} for $j=1$ and \revise{10} for $j=2,3$.
\revise{For the RK method, the weight $w_{k,j}$ is close to these upper bounds at every $k$, which implies that the numerical solution is sufficiently accurate compared to the observation noise variance over the observation period.
On the other hand, the weight is significantly smaller than the upper bounds for the Euler methods, especially when the step size is small.}
In this way, the IRLS algorithms provide information about the reliablity of the numerical solution as a byproduct.

% fss = 1, 20, 100
\begin{figure}[htbp]
\centering
\begin{tabular}{ll}
\
\begin{tikzpicture}
\tikzstyle{every node}=[]
\begin{axis}[width=7cm,%restrict x to domain = 0:69,
xmax=2,xmin=0,
ymax=2.1, ymin = 0,
%%	restrict y to domain=-0.8:1.1, ymax=1.1,ymin=-0.8,
xlabel={$t$},ylabel={weight for $x_1$},
ylabel near ticks,
legend entries={Euler ($\Delta t = 0.005$),Euler ($\Delta t = 0.001$),RK ($\Delta t = 0.005$)},
legend style={legend cell align=left,draw=none,fill=white,fill opacity=0.8,text opacity=1,legend pos=south east,},
	]
%\node[above right] at (axis cs: -2,-1) {$h=0.1$};
%\node[above left] at (axis cs: 0.5,-1.05) {$60$ steps};
\addplot[thick, dotted, color=red,
filter discard warning=false, unbounded coords=discard
] table {
  0.00   2.00000000
  0.19   2.00000000
  0.20   1.32957473
  1.88   1.32957473
  1.89   1.22071978
  1.97   1.22071978
  1.98   1.12895459
  1.99   0.88587109
  2.00   0.88587109
};
\addplot[thick, densely dashed, color=blue,
filter discard warning=false, unbounded coords=discard
] table {
  0.00   2.00000000
  0.19   2.00000000
  0.20   1.95622504
  1.89   1.95622504
  1.90   1.80686730
  2.00   1.80686730
};
\addplot[thick, color=magenta,
filter discard warning=false, unbounded coords=discard
] table {
  0.00   2.00000000
  1.90   2.00000000
  1.91   1.74324223
  2.00   1.74324223
};
\end{axis}
\end{tikzpicture} 
& 
% \hspace{-5em}
\begin{tikzpicture}
\tikzstyle{every node}=[]
\begin{axis}[width=7cm,%restrict x to domain = 0:69,
xmax=2,xmin=0,
ymax=11,ymin=0,
%%	restrict y to domain=-0.8:1.1, ymax=1.1,ymin=-0.8,
xlabel={$t$},ylabel={weight for $x_2$},
ylabel near ticks,
	]
%\node[above right] at (axis cs: -2,-1) {$h=0.1$};
%\node[above left] at (axis cs: 0.5,-1.05) {$60$ steps};
\addplot[thick, dotted, color=red,
filter discard warning=false, unbounded coords=discard
] table {
  0.00  10.00000000
  0.01   9.14049272
  0.02   5.45929889
  0.03   1.88542123
  1.84   1.88542123
  1.85   0.92368629
  1.86   0.86337900
  1.87   0.73390796
  1.94   0.73390796
  1.95   0.55785239
  1.96   0.55785239
  1.97   0.37787499
  1.98   0.34654434
  1.99   0.34557729
  2.00   0.25074819
};
\addplot[thick, densely dashed, color=blue,
filter discard warning=false, unbounded coords=discard
] table {
  0.00  10.00000000
  0.09  10.00000000
  0.10   8.57071219
  1.84   8.57071219
  1.85   5.94868262
  1.86   5.92937586
  1.96   5.92937586
  1.97   5.69450694
  1.99   5.69450694
  2.00   2.64860956
};
\addplot[thick, color=magenta,
filter discard warning=false, unbounded coords=discard
] table {
  0.00  10.00000000
  1.72  10.00000000
  1.73   9.81306863
  2.00   9.81306863
};
\end{axis}
\end{tikzpicture} 
\\
\begin{tikzpicture}
\tikzstyle{every node}=[]
\begin{axis}[width=7cm,%restrict x to domain = 0:69,
xmax=2,xmin=0,
ymax=11,ymin=0,
%%	restrict y to domain=-0.8:1.1, ymax=1.1,ymin=-0.8,
xlabel={$t$},ylabel={weight for $x_3$},
ylabel near ticks,
% legend entries={Euler,Heun,RK},
% legend style={legend cell align=left,draw=none,fill=white,fill opacity=0.8,text opacity=1,
% legend pos=outer north east,
% },
%legend style={legend cell align=left,draw=none,fill=none,at={(0.4,0.4)},anchor=north west}
	]
%\node[above right] at (axis cs: -2,-1) {$h=0.1$};
%\node[above left] at (axis cs: 0.5,-1.05) {$60$ steps};
\addplot[thick, dotted, color=red,
filter discard warning=false, unbounded coords=discard
] table {
  0.00   5.20262528
  0.35   5.20262528
  0.36   2.32543852
  0.56   2.32543852
  0.57   0.97159900
  2.00   0.97159900
};
\addplot[thick, densely dashed, color=blue,
filter discard warning=false, unbounded coords=discard
] table {
  0.00   7.98700898
  2.00   7.98700898
};
\addplot[thick, color=magenta,
filter discard warning=false, unbounded coords=discard
] table {
  0.00   9.59958438
  1.98   9.59958438
  1.99   8.41914143
  2.00   8.41914143

};
\end{axis}
\end{tikzpicture} 
&
\end{tabular}

\caption{
Estimated weights $w_{k,j}^{(20)}$ by IRLS(20) (Algorithm~\ref{alg:irlsL}) for the Lorenz system.
% Behaviour of Algorithm~\ref{alg:irls1} for the Lorenz  system.
% These figures show the weights for each variable after 20 iterations.
}
\label{fig:Lorenz_estimation2}
\end{figure}
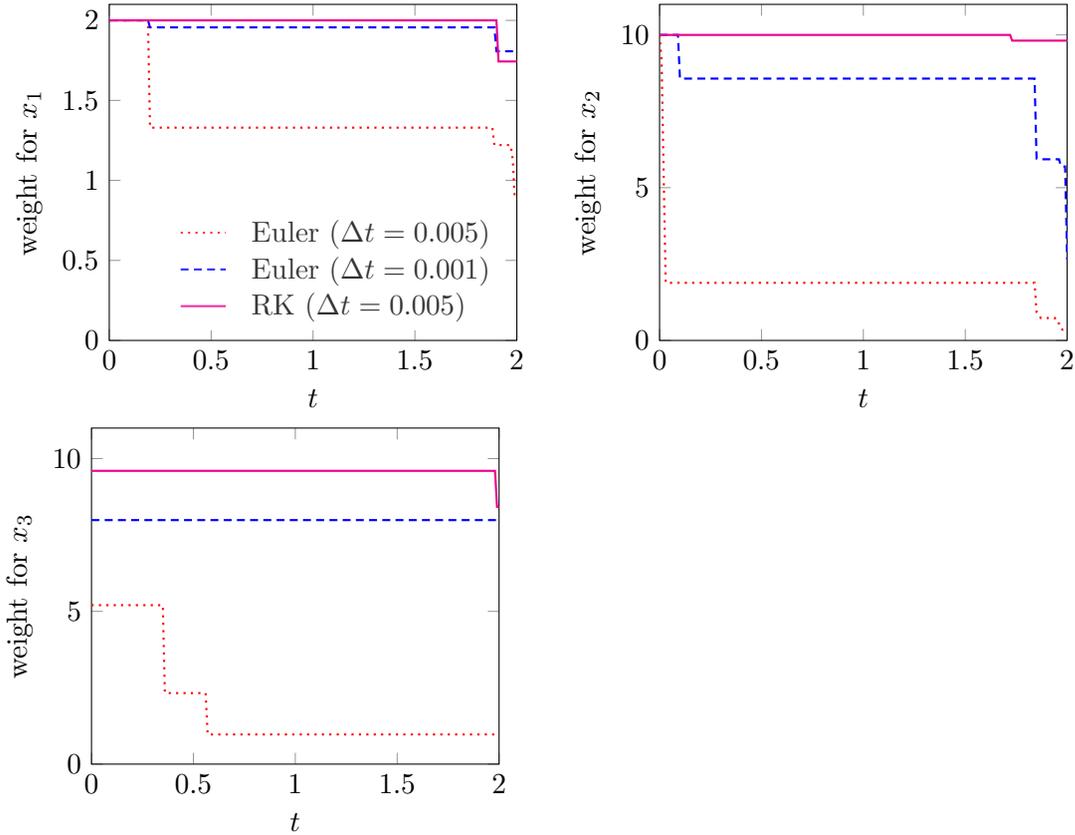

\revisee{
As discussed in \cref{rem:conrad}, we can consider another version of the IRLS algorithm that employs the probabilistic ODE solver \cite{cg17} instead of PAVA for updating the weights.
Whereas this method does not require the monotonicity condition \cref{assmp:sigma} on the discretization error variance, it takes much more computational cost than PAVA.
\cref{fig:Lorenz_comparison_conrad} compares estimation accuracy of these two types of IRLS algorithms\footnote{Since the IRLS algorithm with probabilistic ODE solver also converged within a few iterations, we show the performance of IRLS(3) for both types.} under the same experimental setting with \cref{fig:Lorenz_comparison1}.
For the probabilistic ODE solver, we used 100 samples and determined the perturbation variance following Section 3.1 of \cite{cg17}\footnote{
We estimated the local error of the numerical integrator by comparing its solution with that obtained by the same integrator with the halved step size. We used 100 samples to obtain a Gaussian approximation of the random solver and 50 samples for Monte Carlo sampling.}.
The IRLS algorithm with PAVA attains comparable estimation accuracy to that with the probabilistic ODE solver. 
We will also compare their performance on discretization error quantification in \cref{subsubsec:Lorenz_uq}.
}

% fss = 1, 20, 100
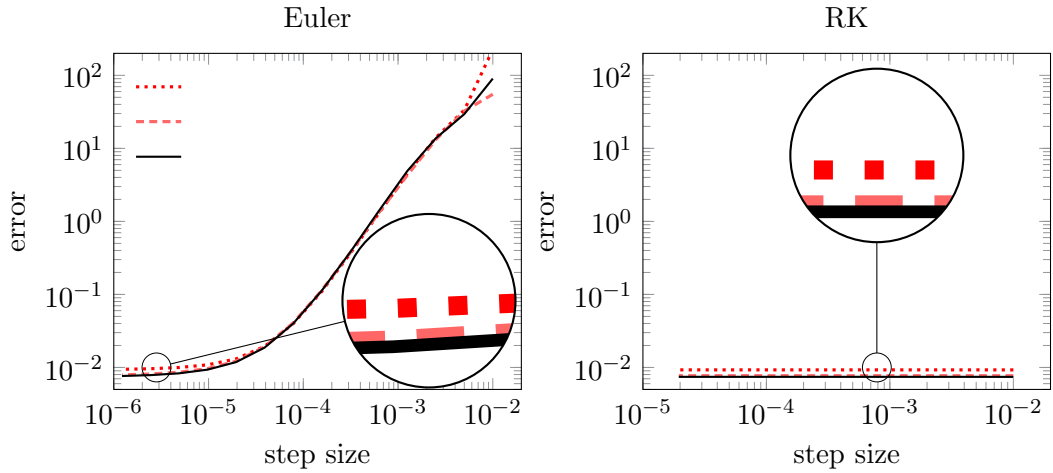
\begin{figure}[htbp]
\centering
\begin{tikzpicture}[spy using outlines=
  {circle, magnification=6, connect spies}]
%\tikzstyle{every node}=[font=\footnotesize]
\tikzstyle{every node}=[]
\begin{axis}[width=7cm,%restrict x to domain = 0:69,
xmax=0.02,xmin=0.000001,
ymax=200, ymin = 0.005,
%%	restrict y to domain=-0.8:1.1, ymax=1.1,ymin=-0.8,
ymode = log,
xmode = log,
xlabel={step size},ylabel={error},
title = {Euler},
ylabel near ticks,
legend entries={conventional,PAVA,probabilistic},
legend style={legend cell align=left,draw=none,fill=white,fill opacity=0,text opacity=1,legend pos=north west,},
	]
%\node[above right] at (axis cs: 0.00005,8) {{\color{red}standard least squares}};
%\node[above right] at (axis cs: 0.0003,1) {{\color{blue}IRLS}};
%\node[above right] at (axis cs: 0.0004,0.03) {{\color{cyan!60} \bf IRLS (1 iteration)}};
\addplot[very thick, color=red, dotted, 
filter discard warning=false, unbounded coords=discard
] table {
0.01 226.3 
0.005 33.41 
0.0025 14.01 
0.00125 4.491 
0.000625 1.285 
0.0003125 0.3625 
0.0001563 0.1095 
7.813e-05 0.03963 
3.906e-05 0.01943 
1.953e-05 0.01309 
9.766e-06 0.01086 
4.883e-06 0.009983 
2.441e-06 0.009605 
1.221e-06 0.009431 
};
\addplot[very thick, densely dashed, color=red!60,
filter discard warning=false, unbounded coords=discard
] table {
	0.01 55.06 
	0.005 32.45 
	0.0025 13.15 
	0.00125 4.281 
	0.000625 1.307 
	0.0003125 0.3732 
	0.0001563 0.1132 
	7.813e-05 0.04067 
	3.906e-05 0.01918 
	1.953e-05 0.01218 
	9.766e-06 0.009633 
	4.883e-06 0.008594 
	2.441e-06 0.008134 
	1.221e-06 0.007919 
};
\addplot[thick, color=black,
filter discard warning=false, unbounded coords=discard
] table {
0.01 90.27 
0.005 29.46 
0.0025 13.73 
0.00125 4.889 
0.000625 1.382 
0.0003125 0.3781 
0.0001563 0.1123 
7.813e-05 0.03986 
3.906e-05 0.01871 
1.953e-05 0.01184 
9.766e-06 0.009336 
4.883e-06 0.008308 
2.441e-06 0.007853 
1.221e-06 0.007641 
};
\coordinate (spypoint) at (axis cs: 0.000002812,0.01);
\coordinate (magnifyglass) at (axis cs: 0.002112, 0.082);
\end{axis}
\spy [black, size=2.3cm] on (spypoint)
  in node[fill=white] at (magnifyglass);
\end{tikzpicture} 
\begin{tikzpicture}[spy using outlines=
{circle, magnification=6, connect spies}]
%\tikzstyle{every node}=[font=\footnotesize]
\tikzstyle{every node}=[]
\begin{axis}[width=7cm,%restrict x to domain = 0:69,
xmax=0.02,xmin=0.00001,
ymax=200, ymin = 0.005,
%%	restrict y to domain=-0.8:1.1, ymax=1.1,ymin=-0.8,
ymode = log,
xmode = log,
xlabel={step size},ylabel={error},
title = {RK},
ylabel near ticks,
]
%\node[above right] at (axis cs: 0.00005,8) {{\color{red}standard least squares}};
%\node[above right] at (axis cs: 0.0003,1) {{\color{blue}IRLS}};
%\node[above right] at (axis cs: 0.0004,0.03) {{\color{cyan!60} \bf IRLS (1 iteration)}};
\addplot[very thick, color=red, dotted, 
filter discard warning=false, unbounded coords=discard
] table {
	0.01 0.009255 
	0.005 0.009266 
	0.0025 0.009266 
	0.00125 0.009266 
	0.000625 0.009266 
	0.0003125 0.009266 
	0.0001563 0.009266 
	7.813e-05 0.009266 
	3.906e-05 0.009266 
	1.953e-05 0.009266 
};
\addplot[very thick, densely dashed, color=red!60,
filter discard warning=false, unbounded coords=discard
] table {
	0.01 0.007704 
	0.005 0.007714 
	0.0025 0.007714 
	0.00125 0.007714 
	0.000625 0.007714 
	0.0003125 0.007714 
	0.0001563 0.007714 
	7.813e-05 0.007714 
	3.906e-05 0.007714 
	1.953e-05 0.007714 
};
\addplot[thick, color=black,
filter discard warning=false, unbounded coords=discard
] table {
0.01 0.007428 
0.005 0.007438 
0.0025 0.007438 
0.00125 0.007438 
0.000625 0.007438 
0.0003125 0.007438 
0.0001563 0.007438 
7.813e-05 0.007438 
3.906e-05 0.007438 
1.953e-05 0.007438 
};
\coordinate (spypoint) at (axis cs: 0.0007812,0.01);
\coordinate (magnifyglass) at (axis cs: 0.0007812, 8);
\end{axis}
\spy [black, size=2.3cm] on (spypoint)
in node[fill=white] at (magnifyglass);
\end{tikzpicture} 

\caption{Errors $l(\hat{\theta},\theta)$ in \cref{Lorenz_err} of IRLS(3) (\cref{alg:irlsL}) with two types of weight updates and the conventional method in \cref{ml_form_approx} for the Lorenz system.  
}
\label{fig:Lorenz_comparison_conrad}
\end{figure}

\revisee{\cref{tab_conf} presents the coverage probability and average length of 95 \% confidence intervals obtained by the method in \cref{subsec:confidence}.
The experimental setting is the same with \cref{fig:Lorenz_estimation1} (RK, $\Delta t = 5.0 \times 10^{-3}$) and we used 100 repetitions.
The bisection method was used with the absolute error tolerance $0.01$.
For all estimands, the coverage probability is approximately equal to the nominal value, and the average length is in the same order of magnitude with the square root of the mean squared error.
Therefore, the method in \cref{subsec:confidence} provides proper uncertainty quantification of the point estimates obtained by the IRLS algorithm.
Although we tried to compute confidence intervals for the conventional method \cref{ml_form_approx} in a similar way, the objective function was rapidly varying and the computation was unstable.
The IRLS algorithm reduces such instability through discretization error quantification.
We also conducted simulation with Euler ($\Delta t = 1.0 \times 10^{-3}$).
However, the coverage probability was very small due to the bias of the point estimates obtained by the IRLS algorithm.
It is an interesting future problem to derive correction methods for such cases, e.g. with a bootstrap procedure. 
Note that the estimated weights by the IRLS algorithm (\cref{fig:Lorenz_estimation2}) already indicated that Euler ($\Delta t = 1.0 \times 10^{-3}$) is not sufficiently accurate in this case.
}
	\begin{table}
	\label{tab_conf}
	\caption{Coverage probability and average length of 95 \% confidence intervals and mean squared errors of IRLS (\cref{alg:irls1}) for the Lorenz system.}
    \centering
	\begin{tabular}{|c|c|c|c|}
		\hline
		estimand & coverage probability & average length & MSE \\ \hline
		$x_1(0)$ & 83 & 0.5976 & $4.22\times 10^{-2}$ \\
		$x_2(0)$ & 88 & 0.3291 & $1.23\times 10^{-2}$ \\
		$x_3(0)$ & 86 & 0.1799 & $3.58\times 10^{-3}$ \\
		$\sigma$ & 84 & 0.1891 & $5.00\times 10^{-3}$ \\
		$\rho$ & 82 & 0.0636 & $6.51\times 10^{-4}$ \\
		$\beta$ & 87 & 0.0149 & $2.68\times 10^{-5}$ \\ \hline
    \end{tabular}
	\end{table}

\subsubsection{Discretization error quantification}
\label{subsubsec:Lorenz_uq}
Now, we confirm that the IRLS algorithms successfully quantify the discretization error.
Here, we regard the output of \textsf{ode45} as $x(t; \theta)$ and fix the time step size for $\tx(t_k; \theta)$ to \revise{$\Delta t = 5.0\times 10^{-3}$}. 
\cref{fig:Lorenz_variables} shows the observation data $y_1,\dots,y_K$ and the numerical solutions $\tx(t; \theta)$ with the true initial state \revise{and true system parameter} of three ODE solvers.
\revise{The discretization error of the Euler method becomes significantly large around $t=1$. 
On the other hand, the numerical solutions with the Heun and RK methods are almost identical and seem to be sufficiently accurate compared to the observation noise variances.
}

\input{Lorenz_variables_revise_h0.005.tex}

\cref{fig:Lorenz_variables_error} plots the estimated weights $\hat{w}_{k,j}$, the square root of the estimated discretization error variance $\hat{\sigma}_{k,j}$, and the actual discretization error $|x_j(t_k; \theta) - \tilde{x}_{j}(t_k; {\theta})|$ \revise{for the Euler method}.
Note that we used the true initial state {and true system parameter} here.
The estimates of discretization error variance quantify the actual discretization error well.
\revise{In particular, the behavior of the rapid growth is well captured.}
\revise{The estimated weights are very small for $t>1$, which indicates that the numerical solutions are no longer reliable.
	On the other hand, the estimated weights for the RK method are also shown in \cref{fig:Lorenz_variables_error}.
	Note that the estimated weights for the Heun method are almost identical to the RK method. 
	The estimated weights for the RK method are almost constant and close to $1/\gamma_j^2$, which implies that the numerical solutions are sufficiently accurate compared to the observation noise variances.
	Similar results are obtained also for the Euler method if the step size is sufficiently small, although the results are not displayed here to save space.
}
In this way, the IRLS algorithms provide information about the accuracy of the numerical solution as a byproduct.

\input{Lorenz_variables_error_revise_h0.005.tex}

\revisee{As discussed in \cref{rem:conrad}, we can consider another version of the IRLS algorithm that employs the probabilistic ODE solver \cite{cg17} instead of PAVA for updating the weights.
\cref{fig:Lorenz_variables_error_conrad} shows discretization error quantification by this IRLS algorithm in the same setting with \cref{fig:Lorenz_variables_error}.
Again, the estimates of discretization error variance quantify the scale of the actual discretization error well.
Note that the probabilistic ODE solver does not require the monotonicity condition \cref{assmp:sigma} on the discretization error variance.
}

\input{Lorenz_variables_error_revise_h0.005_conrad.tex}

\subsubsection{Unknown observation noise variances}\label{sec:unknown}
\revise{
In practice, the observation noise variances $\gamma_j^2$ may be unknown.
As discussed in \cref{rem_unknown}, the proposed method works even in such cases by using lower bounds $\tilde{\gamma}_j^2$ instead of $\gamma_j^2$.
\cref{fig:Lorenz_unknown} shows the results of the proposed method, where the lower bound $\tilde{\Gamma} = \diag ( 0.001,0.001,0.001)$ is used instead of the true observation noise variances $\Gamma = \diag ( 0.5, 0.1, 0.1)$ and the RK method with step size $\Delta t = 5.0 \times 10^{-3}$ is employed.
The estimation accuracy is almost the same with the case of using the true observation noise variances shown in
\cref{fig:Lorenz_estimation1}. 
Also, the estimated weights around $t=0$ are close to $1/\gamma_j^2$, which implies that rough estimates of $\gamma_j^2$ are also obtained.
}

\input{Lorenz_unknown_weight.tex}

\subsubsection{Estimation from limited observation}
\label{subsubsec:Lorenz_limited}
\revise{We consider a more challenging case where both the observation period and observed variables are limited.
Specifically, we restrict the observation data in \cref{fig:Lorenz_variables} to only the first and second variables in the period $[0.1,2]$.
\cref{fig:Lorenz_period_variable_restriction} shows the results of the proposed method, where the RK method with step size $\Delta t = 5.0 \times 10^{-3}$ is employed.
The proposed method still converges in a few iterations, and the estimated weights show a qualitatively similar behavior to \cref{fig:Lorenz_estimation2}.
Thus, the proposed method works well even when only limited observation is available.
}

\begin{figure}[htbp]
\centering
\begin{tikzpicture}
\tikzstyle{every node}=[]
\begin{axis}[width=7cm,%restrict x to domain = 0:69,
xmax=20,xmin=0,
%ymax=100,
%%	restrict y to domain=-0.8:1.1, 
%ymax=300,ymin=0,
ymode = log,
xlabel={iteration count},ylabel={error},
ylabel near ticks,
	]
%\node[above right] at (axis cs: -2,-1) {$h=0.1$};
%\node[above left] at (axis cs: 0.5,-1.05) {$60$ steps};
\addplot[thick, 
filter discard warning=false, unbounded coords=discard
] table {
0 277.74757823
1   0.22760019
2   0.01373339
3   0.01335074
4   0.01335063
5   0.01335063
6   0.01335063
7   0.01335063
8   0.01335063
9   0.01335063
10   0.01335063
11   0.01335063
12   0.01335063
13   0.01335063
14   0.01335063
15   0.01335063
16   0.01335063
17   0.01335063
18   0.01335063
19   0.01335063
20   0.01335063
};
\end{axis}
\end{tikzpicture} 
\begin{tikzpicture}
\tikzstyle{every node}=[]
\begin{axis}[width=7cm,%restrict x to domain = 0:69,
xmax=2,xmin=0,
ymax=10.5, ymin = 0,
%%	restrict y to domain=-0.8:1.1, ymax=1.1,ymin=-0.8,
xlabel={$t$},ylabel={weight},
ylabel near ticks,
legend entries={$x_1$,$x_2$,$x_3$},
legend style={legend cell align=left,draw=none,fill=white,fill opacity=0.8,text opacity=1,at={(0.3,0.4)},anchor=east},
	]
%\node[above right] at (axis cs: -2,-1) {$h=0.1$};
%\node[above left] at (axis cs: 0.5,-1.05) {$60$ steps};
\addplot[thick, 
filter discard warning=false, unbounded coords=discard
] table {
  0.10   2.00000000
  0.11   2.00000000
  0.12   2.00000000
  0.13   2.00000000
  0.14   2.00000000
  0.15   2.00000000
  0.16   2.00000000
  0.17   2.00000000
  0.18   2.00000000
  0.19   2.00000000
  0.20   2.00000000
  0.21   2.00000000
  0.22   2.00000000
  0.23   2.00000000
  0.24   2.00000000
  0.25   2.00000000
  0.26   2.00000000
  0.27   2.00000000
  0.28   2.00000000
  0.29   2.00000000
  0.30   2.00000000
  0.31   2.00000000
  0.32   2.00000000
  0.33   2.00000000
  0.34   2.00000000
  0.35   2.00000000
  0.36   2.00000000
  0.37   2.00000000
  0.38   2.00000000
  0.39   2.00000000
  0.40   2.00000000
  0.41   2.00000000
  0.42   2.00000000
  0.43   2.00000000
  0.44   2.00000000
  0.45   2.00000000
  0.46   2.00000000
  0.47   2.00000000
  0.48   2.00000000
  0.49   2.00000000
  0.50   2.00000000
  0.51   2.00000000
  0.52   2.00000000
  0.53   2.00000000
  0.54   2.00000000
  0.55   2.00000000
  0.56   2.00000000
  0.57   2.00000000
  0.58   2.00000000
  0.59   2.00000000
  0.60   2.00000000
  0.61   2.00000000
  0.62   2.00000000
  0.63   2.00000000
  0.64   2.00000000
  0.65   2.00000000
  0.66   2.00000000
  0.67   2.00000000
  0.68   2.00000000
  0.69   2.00000000
  0.70   2.00000000
  0.71   2.00000000
  0.72   2.00000000
  0.73   2.00000000
  0.74   2.00000000
  0.75   2.00000000
  0.76   2.00000000
  0.77   2.00000000
  0.78   2.00000000
  0.79   2.00000000
  0.80   2.00000000
  0.81   2.00000000
  0.82   2.00000000
  0.83   2.00000000
  0.84   2.00000000
  0.85   2.00000000
  0.86   2.00000000
  0.87   2.00000000
  0.88   2.00000000
  0.89   2.00000000
  0.90   2.00000000
  0.91   2.00000000
  0.92   2.00000000
  0.93   2.00000000
  0.94   2.00000000
  0.95   2.00000000
  0.96   2.00000000
  0.97   2.00000000
  0.98   2.00000000
  0.99   2.00000000
  1.00   2.00000000
  1.01   2.00000000
  1.02   2.00000000
  1.03   2.00000000
  1.04   2.00000000
  1.05   2.00000000
  1.06   2.00000000
  1.07   2.00000000
  1.08   2.00000000
  1.09   2.00000000
  1.10   2.00000000
  1.11   2.00000000
  1.12   2.00000000
  1.13   2.00000000
  1.14   2.00000000
  1.15   2.00000000
  1.16   2.00000000
  1.17   2.00000000
  1.18   2.00000000
  1.19   2.00000000
  1.20   2.00000000
  1.21   2.00000000
  1.22   2.00000000
  1.23   2.00000000
  1.24   2.00000000
  1.25   2.00000000
  1.26   2.00000000
  1.27   2.00000000
  1.28   2.00000000
  1.29   2.00000000
  1.30   2.00000000
  1.31   2.00000000
  1.32   2.00000000
  1.33   2.00000000
  1.34   2.00000000
  1.35   2.00000000
  1.36   2.00000000
  1.37   2.00000000
  1.38   2.00000000
  1.39   2.00000000
  1.40   2.00000000
  1.41   2.00000000
  1.42   2.00000000
  1.43   2.00000000
  1.44   2.00000000
  1.45   2.00000000
  1.46   2.00000000
  1.47   2.00000000
  1.48   2.00000000
  1.49   2.00000000
  1.50   2.00000000
  1.51   2.00000000
  1.52   2.00000000
  1.53   2.00000000
  1.54   2.00000000
  1.55   2.00000000
  1.56   2.00000000
  1.57   2.00000000
  1.58   2.00000000
  1.59   2.00000000
  1.60   2.00000000
  1.61   2.00000000
  1.62   2.00000000
  1.63   2.00000000
  1.64   2.00000000
  1.65   2.00000000
  1.66   2.00000000
  1.67   2.00000000
  1.68   2.00000000
  1.69   2.00000000
  1.70   2.00000000
  1.71   2.00000000
  1.72   2.00000000
  1.73   2.00000000
  1.74   2.00000000
  1.75   2.00000000
  1.76   2.00000000
  1.77   2.00000000
  1.78   2.00000000
  1.79   2.00000000
  1.80   2.00000000
  1.81   2.00000000
  1.82   2.00000000
  1.83   2.00000000
  1.84   2.00000000
  1.85   2.00000000
  1.86   2.00000000
  1.87   2.00000000
  1.88   2.00000000
  1.89   2.00000000
  1.90   2.00000000
  1.91   1.73995429
  1.92   1.73995429
  1.93   1.73995429
  1.94   1.73995429
  1.95   1.73995429
  1.96   1.73995429
  1.97   1.73995429
  1.98   1.73995429
  1.99   1.73995429
  2.00   1.73995429
%   0.10   2.00000000
%   0.19   2.00000000
%   0.20   1.99736804
%   1.90   1.99736804
%   1.91   1.74961103
%   2.00   1.74961103
};
\addplot[thick, densely dashed,
filter discard warning=false, unbounded coords=discard
] table {
  0.10  10.00000000
  0.11  10.00000000
  0.12  10.00000000
  0.13  10.00000000
  0.14  10.00000000
  0.15  10.00000000
  0.16  10.00000000
  0.17  10.00000000
  0.18  10.00000000
  0.19  10.00000000
  0.20  10.00000000
  0.21  10.00000000
  0.22  10.00000000
  0.23  10.00000000
  0.24  10.00000000
  0.25  10.00000000
  0.26  10.00000000
  0.27  10.00000000
  0.28  10.00000000
  0.29  10.00000000
  0.30  10.00000000
  0.31  10.00000000
  0.32  10.00000000
  0.33  10.00000000
  0.34  10.00000000
  0.35  10.00000000
  0.36  10.00000000
  0.37  10.00000000
  0.38  10.00000000
  0.39  10.00000000
  0.40  10.00000000
  0.41  10.00000000
  0.42  10.00000000
  0.43  10.00000000
  0.44  10.00000000
  0.45  10.00000000
  0.46  10.00000000
  0.47  10.00000000
  0.48  10.00000000
  0.49  10.00000000
  0.50  10.00000000
  0.51  10.00000000
  0.52  10.00000000
  0.53  10.00000000
  0.54  10.00000000
  0.55  10.00000000
  0.56  10.00000000
  0.57  10.00000000
  0.58  10.00000000
  0.59  10.00000000
  0.60  10.00000000
  0.61  10.00000000
  0.62  10.00000000
  0.63  10.00000000
  0.64  10.00000000
  0.65  10.00000000
  0.66  10.00000000
  0.67  10.00000000
  0.68  10.00000000
  0.69  10.00000000
  0.70  10.00000000
  0.71  10.00000000
  0.72  10.00000000
  0.73  10.00000000
  0.74  10.00000000
  0.75  10.00000000
  0.76  10.00000000
  0.77  10.00000000
  0.78  10.00000000
  0.79  10.00000000
  0.80  10.00000000
  0.81  10.00000000
  0.82  10.00000000
  0.83  10.00000000
  0.84  10.00000000
  0.85  10.00000000
  0.86  10.00000000
  0.87  10.00000000
  0.88  10.00000000
  0.89  10.00000000
  0.90  10.00000000
  0.91  10.00000000
  0.92  10.00000000
  0.93  10.00000000
  0.94  10.00000000
  0.95  10.00000000
  0.96  10.00000000
  0.97  10.00000000
  0.98  10.00000000
  0.99  10.00000000
  1.00  10.00000000
  1.01  10.00000000
  1.02  10.00000000
  1.03  10.00000000
  1.04  10.00000000
  1.05  10.00000000
  1.06  10.00000000
  1.07  10.00000000
  1.08  10.00000000
  1.09  10.00000000
  1.10  10.00000000
  1.11  10.00000000
  1.12  10.00000000
  1.13  10.00000000
  1.14  10.00000000
  1.15  10.00000000
  1.16  10.00000000
  1.17  10.00000000
  1.18  10.00000000
  1.19  10.00000000
  1.20  10.00000000
  1.21  10.00000000
  1.22  10.00000000
  1.23  10.00000000
  1.24  10.00000000
  1.25  10.00000000
  1.26  10.00000000
  1.27  10.00000000
  1.28  10.00000000
  1.29  10.00000000
  1.30  10.00000000
  1.31  10.00000000
  1.32  10.00000000
  1.33  10.00000000
  1.34  10.00000000
  1.35  10.00000000
  1.36  10.00000000
  1.37  10.00000000
  1.38  10.00000000
  1.39  10.00000000
  1.40  10.00000000
  1.41  10.00000000
  1.42  10.00000000
  1.43  10.00000000
  1.44  10.00000000
  1.45  10.00000000
  1.46  10.00000000
  1.47  10.00000000
  1.48  10.00000000
  1.49  10.00000000
  1.50  10.00000000
  1.51  10.00000000
  1.52  10.00000000
  1.53  10.00000000
  1.54  10.00000000
  1.55  10.00000000
  1.56  10.00000000
  1.57  10.00000000
  1.58  10.00000000
  1.59  10.00000000
  1.60  10.00000000
  1.61  10.00000000
  1.62  10.00000000
  1.63  10.00000000
  1.64  10.00000000
  1.65  10.00000000
  1.66  10.00000000
  1.67  10.00000000
  1.68  10.00000000
  1.69  10.00000000
  1.70  10.00000000
  1.71  10.00000000
  1.72  10.00000000
  1.73   9.82874623
  1.74   9.82874623
  1.75   9.82874623
  1.76   9.82874623
  1.77   9.82874623
  1.78   9.82874623
  1.79   9.82874623
  1.80   9.82874623
  1.81   9.82874623
  1.82   9.82874623
  1.83   9.82874623
  1.84   9.82874623
  1.85   9.82874623
  1.86   9.82874623
  1.87   9.82874623
  1.88   9.82874623
  1.89   9.82874623
  1.90   9.82874623
  1.91   9.82874623
  1.92   9.82874623
  1.93   9.82874623
  1.94   9.82874623
  1.95   9.82874623
  1.96   9.82874623
  1.97   9.82874623
  1.98   9.82874623
  1.99   9.82874623
  2.00   9.82874623
%   0.10  10.00000000
%   1.72  10.00000000
%   1.73   9.79164294
%   2.00   9.79164294
};
\end{axis}
\end{tikzpicture} 

\caption{
Error in \eqref{Lorenz_err} of IRLS (Algorithm~\ref{alg:irls1}) and estimated weights at the 20th iteration for the Lorenz system, where only the first and second variables are observed in the period $[0.1,2]$ and the RK method with step size $\Delta t = 5.0 \times 10^{-3}$ is employed.
}
\label{fig:Lorenz_period_variable_restriction}
\end{figure}
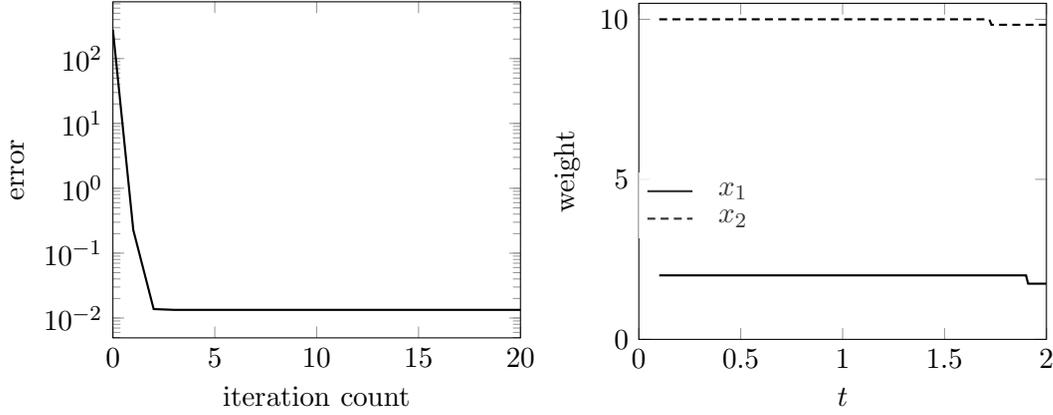

\subsection{FitzHugh--Nagumo model}
Here, we consider the FitzHugh--Nagumo model \cite{f61,nay62}:
\begin{align*}
\frac{\rmd}{\rmd t} \begin{bmatrix}
V \\ R
\end{bmatrix}
=
\begin{bmatrix} 
c\paren*{V-\cfrac{V^3}{3}+R}\\
-\cfrac{1}{c} \paren*{V-a+bR}
\end{bmatrix},
\quad
\begin{bmatrix}
V(0) \\ R(0)
\end{bmatrix}
= \begin{bmatrix}
-1 \\-1
\end{bmatrix}.
\end{align*}

We consider estimation of the system parameter $\theta = [a,b,c] = [0.2,0.2,3]$ from the observation of $V$.
Namely, the observation matrix $H$ in \cref{obs_model2} is $H=[1, 0]$.
The observation is taken at $t_k=(k-1)h$ for $k=1,\dots,K$, where $h=0.2$ and $K=201$, and the observation noise variance is set to $\Gamma = 0.01$.

\cref{fig:FN_estimation1} (top left, top right) plots the objective function $g(\theta^{(l)},w^{(l)})$ in \cref{g_def} and squared error $\| \hat{\theta}^{(l)}-\theta \|^2$ with respect to the iteration count $l$ for different step sizes. 
The initial guess was set to $\theta^{(0)} = [1,1,1]$.
Similarly to the Lorenz system, estimation accuracy is better for smaller step sizes.
Also, \cref{fig:FN_estimation1} (bottom) shows the estimates of weights ${w}_{k,j}^{(20)}$, which indicates that the numerical solution with smaller step size is more reliable.

\input{FN_estimation1.tex}

\cref{fig:FN_comparison1} plots the squared errors of the IRLS algorithms and the conventional method with respect to the step size.
The IRLS algorithms have better estimation accuracy when the step sizes are smaller than $10^{-2}$, whereas both methods attain almost the same estimation accuracy for relatively large step sizes.
Since the plots for IRLS(2), IRLS(3) and IRLS almost overlap each other, IRLS algorithms are considered to converge in two iterations.

% fss = 1, 20, 100
\begin{figure}[htbp]
\centering
\begin{tikzpicture}[spy using outlines=
  {circle, magnification=6, connect spies}]
%\tikzstyle{every node}=[font=\footnotesize]
\tikzstyle{every node}=[]
\begin{axis}[width=7cm,%restrict x to domain = 0:69,
xmax=0.5,xmin=0.0001,
ymax=0.2, ymin = 0.001,
%%	restrict y to domain=-0.8:1.1, ymax=1.1,ymin=-0.8,
ymode = log,
xmode = log,
xlabel={step size},ylabel={squared error},
ylabel near ticks,
legend entries={conventional,IRLS(1),IRLS(2),IRLS(3),IRLS},
legend style={legend cell align=left,draw=none,fill=white,fill opacity=0.8,text opacity=1,legend pos=north west,},
	]
%\node[above right] at (axis cs: 0.00005,8) {{\color{red}standard least squares}};
%\node[above right] at (axis cs: 0.0003,1) {{\color{blue}IRLS}};
%\node[above right] at (axis cs: 0.0004,0.03) {{\color{cyan!60} \bf IRLS (1 iteration)}};
\addplot[very thick, dotted, color=red,
filter discard warning=false, unbounded coords=discard
] table {
0.2	0.127765085
0.1	0.044602683
0.05	0.016498242
0.025	0.007880639
0.0125	0.00499123
0.00625	0.00390913
0.003125	0.003459543
0.0015625	0.003257686
0.00078125	0.003162498
0.000390625	0.00311634
};
\addplot[very thick, dashed, color=cyan!60,
filter discard warning=false, unbounded coords=discard
] table {
0.2	0.112835862
0.1	4.35E-02
0.05	1.67E-02
0.025	7.10E-03
0.0125	3.64E-03
0.00625	2.33E-03
0.003125	0.001784956
0.0015625	0.001544615
0.00078125	0.001432307
0.000390625	0.001378142
};
\addplot[very thick, densely dashdotted, color=green!60,
filter discard warning=false, unbounded coords=discard
] table {
0.2	0.12774683
0.1	4.43E-02
0.05	1.57E-02
0.025	7.02E-03
0.0125	4.16E-03
0.00625	3.10E-03
0.003125	0.00266218
0.0015625	0.002467294
0.00078125	0.002375741
0.000390625	0.002331438
};
\addplot[very thick, densely dashed, color=red!60,
filter discard warning=false, unbounded coords=discard
] table {
0.2	0.127750379
0.1	4.43E-02
0.05	1.56E-02
0.025	7.02E-03
0.0125	4.18E-03
0.00625	3.12E-03
0.003125	0.002692853
0.0015625	0.002500433
0.00078125	0.002410093
0.000390625	0.002366391
};
\addplot[thick, color=blue,
filter discard warning=false, unbounded coords=discard
] table {
0.2	1.28E-01
0.1	4.43E-02
0.05	1.56E-02
0.025	7.02E-03
0.0125	4.18E-03
0.00625	0.003125399
0.003125	0.002693864
0.0015625	0.002501535
0.00078125	0.00241124
0.000390625	0.002367525
};
\coordinate (spypoint) at (axis cs: 0.0007812,0.002367525);
\coordinate (magnifyglass) at (axis cs: 0.5, 0.01);
\end{axis}
\spy [black, size=2.5cm] on (spypoint)
  in node[fill=white] at (magnifyglass);
\end{tikzpicture}

\caption{Comparison of the squared error $\| \hat{\theta}-\theta \|^2$ of IRLS (Algorithm~\ref{alg:irls1}), IRLS($L$) (Algorithm~\ref{alg:irlsL}) with several values of $L$, and the conventional method in \eqref{ml_form_approx} for the FitzHugh--Nagumo model. 
The plots for IRLS(2), IRLS(3), IRLS almost overlap each other.
}
\label{fig:FN_comparison1}
\end{figure}
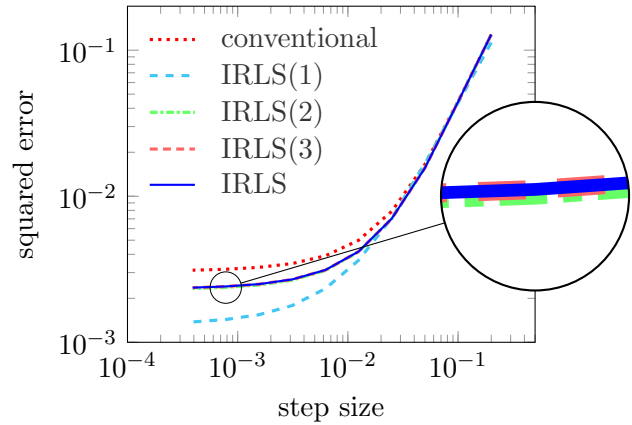

\subsection{Kepler equation}
\label{subsec:Kepler}
Finally, we consider initial value estimation of the Kepler equation:
\begin{align*}
    \frac{\rmd}{\rmd t} 
    \begin{bmatrix}
        q_1 \\ q_2 \\ p_1 \\ p_2
    \end{bmatrix}
    =
    \begin{bmatrix}
        p_1 \\ p_2 \\ -\cfrac{q_1}{(q_1^2+q_2^2)^{3/2}} \\ -\cfrac{q_2}{(q_1^2+q_2^2)^{3/2}}
    \end{bmatrix},
    \quad
    \begin{bmatrix}
        q_1(0) \\ q_2 (0) \\ p_1 (0) \\ p_2 (0)
    \end{bmatrix}
    = \theta.
\end{align*}

We consider \revise{estimation of $\theta = [1-e,0,0,\sqrt{(1+e)/(1-e)}]^{\top}=[0.4,0,0,2]^\top$} with the eccentricity $e=0.6$ from the observations of $q_1$ and $q_2$.
Namely, the observation matrix $H$ in \cref{obs_model2} is 
\begin{align*}
    H=\begin{bmatrix}1 & 0 & 0 & 0 \\ 0 & 1 & 0 & 0 \end{bmatrix}.    
\end{align*}
The observation is taken at $t_k=(k-1)h$ for $k=1,\dots,K$, where $h=0.2$ and $K=101$, and the observation noise variance is set to $\Gamma = \diag (1\times10^{-4},1\times10^{-4})$.

The Kepler equation is a Hamiltonian system, and symplectic integrators are often employed for solving Hamiltonian systems in practical computations~\cite{hl06}.
Thus, we use the Störmer--Verlet method for the numerical solution $\tx(t;\theta)$ (see \cref{app_sv} for details). 

\cref{fig:Kepler_comparison1} plots the squared errors of the IRLS algorithms and the conventional method with respect to the step size.
We set the initial guess to $\theta^{(0)}=[0.5,0.05,-0.05,2.5]^\top$.
Similar to the left figure in \cref{fig:Lorenz_comparison1}, the IRLS algorithms have better estimation accuracy than the conventional method, even with a single iteration ($L=1$).
\revise{
Note that the result of the conventional approach is not improved even when the step size is small.
It seems that the iteration of the optimization solver is stagnated at a local minimum.
This behavior may be improved if another optimization solver is employed.
But we would like to emphasize that the proposed algorithm gives much better results with the same optimization solver.
}

% fss = 1, 20, 100
\begin{figure}[htbp]
\centering
\begin{tikzpicture}
%\tikzstyle{every node}=[font=\footnotesize]
\tikzstyle{every node}=[]
\begin{axis}[width=7cm,%restrict x to domain = 0:69,
xmax=0.5,xmin=0.0001,
ymax=100, ymin = 0.000001,
%%	restrict y to domain=-0.8:1.1, ymax=1.1,ymin=-0.8,
ymode = log,
xmode = log,
xlabel={step size},ylabel={squared error},
ylabel near ticks,
legend entries={conventional,IRLS(1),IRLS(2),IRLS(3),IRLS},
legend style={legend cell align=left,draw=none,fill=white,fill opacity=0.8,text opacity=1,legend pos=outer north east,},
	]
%\node[above right] at (axis cs: 0.00005,8) {{\color{red}standard least squares}};
%\node[above right] at (axis cs: 0.0003,1) {{\color{blue}IRLS}};
%\node[above right] at (axis cs: 0.0004,0.03) {{\color{cyan!60} \bf IRLS (1 iteration)}};
\addplot[very thick, dotted, color=red,
filter discard warning=false, unbounded coords=discard
] table {
0.2	10.98754166
0.1	9.258978074
0.05	0.758723382
0.025	1.54E+00
0.0125	4.20E-01
0.00625	4.36E-01
0.003125	1.653011483
0.0015625	1.967251717
0.00078125	1.768226139
0.000390625	2.637962882
};
\addplot[very thick, dashed, color=cyan!60,
filter discard warning=false, unbounded coords=discard
] table {
0.2	0.018121209
0.1	0.004282475
0.05	0.002424203
0.025	0.002046015
0.0125	0.001956553
0.00625	0.001934538
0.003125	0.001929055
0.0015625	0.001927686
0.00078125	0.001927343
0.000390625	0.001927259
};
\addplot[very thick, densely dashdotted, color=green!60,
filter discard warning=false, unbounded coords=discard
] table {
0.2	0.013639745
0.1	8.08561E-05
0.05	0.000158522
0.025	6.32847E-06
0.0125	2.51048E-05
0.00625	4.94406E-05
0.003125	5.71955E-05
0.0015625	5.92451E-05
0.00078125	5.97655E-05
0.000390625	5.98956E-05
};
\addplot[very thick, densely dashed, color=red!60,
filter discard warning=false, unbounded coords=discard
] table {
0.2	0.012727184
0.1	0.000864733
0.05	0.001114307
0.025	6.92794E-06
0.0125	2.51716E-05
0.00625	4.94632E-05
0.003125	5.72083E-05
0.0015625	5.92556E-05
0.00078125	5.97753E-05
0.000390625	5.99053E-05
};
\addplot[thick, color=blue,
filter discard warning=false, unbounded coords=discard
] table {
0.2	0.012199018
0.1	0.014650174
0.05	0.00111863
0.025	2.35114E-05
0.0125	2.45661E-05
0.00625	4.88446E-05
0.003125	5.65903E-05
0.0015625	5.86382E-05
0.00078125	5.91578E-05
0.000390625	5.94022E-05
};
\end{axis}
\end{tikzpicture} 

\caption{Comparison of the squared error $\| \hat{\theta}-\theta \|^2$ of IRLS (Algorithm~\ref{alg:irls1}), IRLS($L$) (Algorithm~\ref{alg:irlsL}) with several values of $L$, and the conventional method in \eqref{ml_form_approx} for the Kepler equation.}
\label{fig:Kepler_comparison1}
\end{figure}
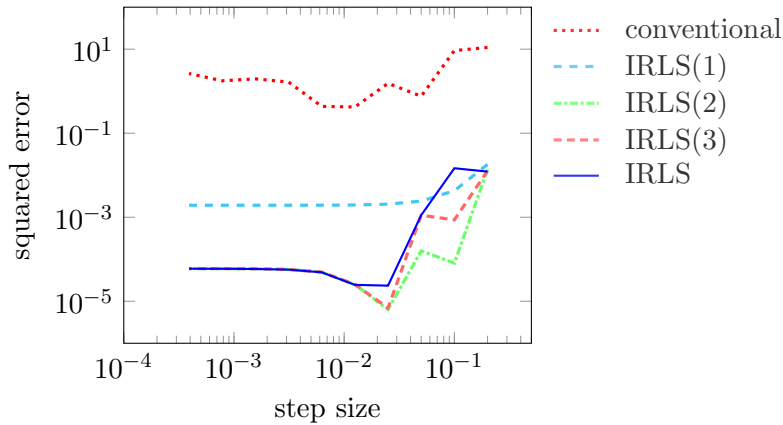

\cref{fig:Kepler_comparison2} shows the numerical solution of the Störmer--Verlet method with the initial state set to the estimate of IRLS(1) or the conventional method.
The trajectory from the estimate of IRLS(1) is fairly close to the observations and well exhibit the elliptic orbit, whereas the one from the estimate of the conventional method does not reproduce the elliptic orbit.

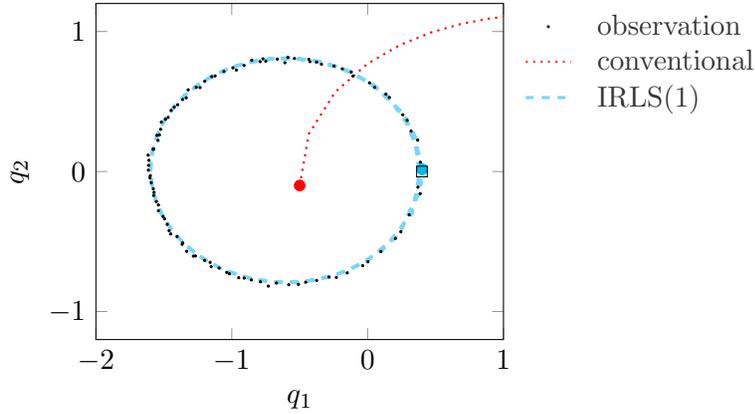
\begin{figure}[htbp]
\centering

\begin{tikzpicture}
%\tikzstyle{every node}=[font=\footnotesize]
\tikzstyle{every node}=[]
\begin{axis}[width=7cm,%restrict x to domain = 0:69,
xmax=1,xmin=-2,
ymax=1.2, ymin = -1.2,
xlabel={$q_1$},ylabel={$q_2$},
ylabel near ticks,
legend entries={observation,conventional,IRLS(1)},
legend style={legend cell align=left,draw=none,fill=white,fill opacity=0.8,text opacity=1,legend pos=outer north east,},
	]
%\node[above right] at (axis cs: 0.00005,8) {{\color{red}standard least squares}};
%\node[above right] at (axis cs: 0.0003,1) {{\color{blue}IRLS}};
%\node[above right] at (axis cs: 0.0004,0.03) {{\color{cyan!60} \bf IRLS (1 iteration)}};
\addplot[only marks,mark options={scale=0.2},
% thick
] table {
    0.3992    0.0160
    0.2909    0.3654
    0.0517    0.6014
   -0.1864    0.7335
   -0.4242    0.7846
   -0.6087    0.7771
   -0.7914    0.7849
   -0.9658    0.7247
   -1.1225    0.6822
   -1.2334    0.5996
   -1.3351    0.5261
   -1.4236    0.4595
   -1.4883    0.3731
   -1.5324    0.2610
   -1.5746    0.1605
   -1.6093    0.0800
   -1.5993   -0.0297
   -1.5778   -0.1229
   -1.5553   -0.2147
   -1.5054   -0.3209
   -1.4612   -0.4223
   -1.3652   -0.5196
   -1.2826   -0.6043
   -1.1591   -0.6491
   -1.0368   -0.7203
   -0.9086   -0.7613
   -0.7305   -0.8185
   -0.5093   -0.8062
   -0.2836   -0.7532
   -0.0363   -0.6723
    0.2049   -0.4786
    0.3872   -0.1581
    0.3683    0.2245
    0.1594    0.5292
   -0.1079    0.6815
   -0.3394    0.7657
   -0.5387    0.8071
   -0.7498    0.7874
   -0.9096    0.7574
   -1.0575    0.7141
   -1.1989    0.6444
   -1.3209    0.5831
   -1.3976    0.4771
   -1.4682    0.3863
   -1.5235    0.3017
   -1.5429    0.2222
   -1.6151    0.1166
   -1.6139    0.0098
   -1.5924   -0.0800
   -1.5752   -0.1708
   -1.5398   -0.2804
   -1.4748   -0.3770
   -1.4025   -0.4645
   -1.3318   -0.5699
   -1.2000   -0.6346
   -1.0933   -0.6898
   -0.9489   -0.7584
   -0.7889   -0.7876
   -0.5808   -0.8086
   -0.3752   -0.7779
   -0.1501   -0.7241
    0.0978   -0.5717
    0.3063   -0.3096
    0.3883    0.0661
    0.2588    0.4052
    0.0142    0.6330
   -0.2475    0.7533
   -0.4329    0.8014
   -0.6648    0.8028
   -0.8528    0.7949
   -0.9915    0.7408
   -1.1471    0.6774
   -1.2617    0.6006
   -1.3610    0.5141
   -1.4459    0.4359
   -1.5183    0.3526
   -1.5506    0.2583
   -1.5788    0.1583
   -1.6126    0.0438
   -1.6028   -0.0449
   -1.5720   -0.1484
   -1.5399   -0.2388
   -1.4940   -0.3378
   -1.4491   -0.4457
   -1.3608   -0.5095
   -1.2676   -0.6031
   -1.1515   -0.6799
   -1.0175   -0.7305
   -0.8554   -0.7740
   -0.6708   -0.7965
   -0.4562   -0.7905
   -0.2300   -0.7561
    0.0024   -0.6443
    0.2443   -0.4303
    0.3705   -0.1118
    0.3200    0.2889
    0.1349    0.5517
   -0.1201    0.7108
   -0.3756    0.7816
   -0.5870    0.8150
   -0.7704    0.8047
};
\addplot[thick, dotted, color=red,
filter discard warning=false, unbounded coords=discard
] table {
   -0.5000   -0.0999
   -0.4325    0.2687
   -0.2402    0.5659
   -0.0017    0.7691
    0.2392    0.9040
    0.4692    0.9939
    0.6853    1.0534
    0.8876    1.0914
    1.0772    1.1136
    1.2551    1.1237
    1.4226    1.1244
    1.5805    1.1174
    1.7297    1.1043
    1.8709    1.0861
    2.0046    1.0635
    2.1315    1.0373
    2.2520    1.0080
    2.3665    0.9759
    2.4754    0.9416
    2.5789    0.9052
    2.6774    0.8670
    2.7710    0.8273
    2.8601    0.7862
    2.9448    0.7439
    3.0253    0.7006
    3.1017    0.6563
    3.1742    0.6112
    3.2430    0.5653
    3.3082    0.5188
    3.3698    0.4718
    3.4280    0.4243
    3.4829    0.3764
    3.5345    0.3281
    3.5830    0.2795
    3.6283    0.2307
    3.6707    0.1817
    3.7101    0.1326
    3.7466    0.0833
    3.7802    0.0340
    3.8111   -0.0154
    3.8392   -0.0647
    3.8646   -0.1140
    3.8873   -0.1632
    3.9074   -0.2123
    3.9248   -0.2612
    3.9397   -0.3100
    3.9520   -0.3586
    3.9618   -0.4069
    3.9691   -0.4550
    3.9739   -0.5029
    3.9763   -0.5503
    3.9761   -0.5975
    3.9735   -0.6443
    3.9685   -0.6907
    3.9611   -0.7367
    3.9512   -0.7822
    3.9389   -0.8272
    3.9242   -0.8718
    3.9071   -0.9157
    3.8876   -0.9592
    3.8656   -1.0020
    3.8412   -1.0442
    3.8144   -1.0858
    3.7851   -1.1266
    3.7534   -1.1667
    3.7192   -1.2061
    3.6825   -1.2446
    3.6433   -1.2823
    3.6016   -1.3191
    3.5573   -1.3549
    3.5104   -1.3898
    3.4610   -1.4237
    3.4088   -1.4564
    3.3541   -1.4880
    3.2965   -1.5185
    3.2363   -1.5476
    3.1732   -1.5754
    3.1073   -1.6018
    3.0385   -1.6267
    2.9667   -1.6499
    2.8918   -1.6716
    2.8139   -1.6914
    2.7328   -1.7093
    2.6484   -1.7251
    2.5606   -1.7388
    2.4694   -1.7501
    2.3747   -1.7589
    2.2762   -1.7649
    2.1740   -1.7681
    2.0678   -1.7679
    1.9574   -1.7643
    1.8428   -1.7568
    1.7238   -1.7451
    1.6000   -1.7286
    1.4714   -1.7068
    1.3376   -1.6790
    1.1984   -1.6445
    1.0535   -1.6020
    0.9027   -1.5505
    0.7456   -1.4882
    0.5820   -1.4129
};
\addplot[very thick, dashed, color=cyan!60, opacity=0.7,
filter discard warning=false, unbounded coords=discard
] table {
    0.4025    0.0156
    0.2879    0.3759
    0.0583    0.6094
   -0.1872    0.7363
   -0.4173    0.7947
   -0.6248    0.8083
   -0.8089    0.7912
   -0.9707    0.7520
   -1.1115    0.6965
   -1.2326    0.6285
   -1.3351    0.5509
   -1.4199    0.4660
   -1.4876    0.3755
   -1.5388    0.2808
   -1.5739    0.1832
   -1.5932    0.0837
   -1.5968   -0.0166
   -1.5847   -0.1167
   -1.5568   -0.2157
   -1.5129   -0.3125
   -1.4525   -0.4058
   -1.3752   -0.4944
   -1.2802   -0.5766
   -1.1667   -0.6504
   -1.0336   -0.7133
   -0.8796   -0.7616
   -0.7033   -0.7904
   -0.5033   -0.7921
   -0.2792   -0.7548
   -0.0347   -0.6585
    0.2106   -0.4689
    0.3846   -0.1486
    0.3587    0.2384
    0.1605    0.5276
   -0.0845    0.6940
   -0.3224    0.7778
   -0.5394    0.8084
   -0.7332    0.8033
   -0.9042    0.7729
   -1.0537    0.7239
   -1.1830    0.6610
   -1.2933    0.5874
   -1.3855    0.5056
   -1.4604    0.4174
   -1.5187    0.3245
   -1.5607    0.2280
   -1.5868    0.1292
   -1.5972    0.0292
   -1.5919   -0.0712
   -1.5708   -0.1708
   -1.5338   -0.2687
   -1.4805   -0.3637
   -1.4105   -0.4547
   -1.3232   -0.5400
   -1.2177   -0.6178
   -1.0931   -0.6859
   -0.9481   -0.7411
   -0.7813   -0.7792
   -0.5913   -0.7938
   -0.3771   -0.7751
   -0.1399   -0.7071
    0.1100   -0.5621
    0.3288   -0.2986
    0.3980    0.0808
    0.2553    0.4236
    0.0196    0.6371
   -0.2239    0.7514
   -0.4501    0.8020
   -0.6535    0.8107
   -0.8338    0.7902
   -0.9922    0.7487
   -1.1299    0.6915
   -1.2481    0.6223
   -1.3480    0.5439
   -1.4302    0.4584
   -1.4956    0.3674
   -1.5445    0.2724
   -1.5775    0.1745
   -1.5946    0.0749
   -1.5960   -0.0254
   -1.5818   -0.1255
   -1.5517   -0.2244
   -1.5055   -0.3208
   -1.4427   -0.4138
   -1.3629   -0.5018
   -1.2653   -0.5833
   -1.1489   -0.6561
   -1.0127   -0.7174
   -0.8552   -0.7636
   -0.6751   -0.7894
   -0.4710   -0.7865
   -0.2427   -0.7421
    0.0046   -0.6341
    0.2463   -0.4255
    0.3964   -0.0848
    0.3347    0.2954
    0.1227    0.5628
   -0.1223    0.7137
   -0.3568    0.7880
   -0.5698    0.8127
   -0.7597    0.8036
};
\addplot[smooth,mark=*,red] plot coordinates {
    (-0.5,-0.0998997723600318)
};
\addplot[smooth,mark=*,cyan!80] plot coordinates {
    (0.4025,    0.0156)
};
\addplot[smooth,mark=square,] plot coordinates {
    (0.4, 0)
};
\end{axis}
\end{tikzpicture} 

\caption{
Numerical solutions of the Störmer--Verlet method (step size $\Delta t=0.02$) with the initial state set to the estimate of IRLS(1) and the conventional method in \eqref{ml_form_approx} for the Kepler equation.
The true initial state is indicated by the square mark. %, and the other two are the estimated initial states. 
}
\label{fig:Kepler_comparison2}
\end{figure}

\section{Concluding remarks}
\label{sec:discussion}

In this paper, we developed a parameter estimation method for ODE models that quantifies the discretization error \revise{based on data}. By modeling the discretization error as random variables, the proposed method alternately updates the discretization error variance and the ODE parameter.
The algorithm for isotonic regression (PAVA) \cite{bb72,rwd88,vE06} is employed for the update of the discretization error variance, whereas the adjoint system and symplectic partitioned Runge--Kutta method \cite{ss16} are used for the update of the ODE parameter. 
Experimental results on several ODE models demonstrated that the proposed method attains robust estimation by successfully quantifying the reliability of numerical solutions.
Also, the proposed method showed better estimation accuracy even when the numerical solution is sufficiently accurate.
Since the proposed method converged in a few iterations in most cases, even one or two iterations is expected to provide a better estimate than the conventional method in practice.

We point out several directions for future work.
While we modeled the discretization error by Gaussian random variables and assumed these random variables to be independent, these assumptions may not be suitable in some cases.
For example, for an ODE system with almost periodic orbit, the discretization error is also considered to be almost periodic.
Thus, it is {important} to extend the proposed method to account for the dependence between discretization error at different time points (also see \cref{rem_brown}).
Also, it is an interesting problem from the viewpoint of numerical analysis to investigate the behavior of the estimates of discretization error variance theoretically. 
\revisee{Finally, the proposed method is expected to be more beneficial for large-scale ODE inverse problems and we leave such practical applications to our future work.}

\appendix
\section{Proof of \cref{thm:pava}}
\label{app_proof}

We prove the following more general statement, of which \cref{thm:pava} is a special case.

\begin{theorem}
Let $\hat{\nu} = (\hat{\nu}_1,\dots,\hat{\nu}_K)$ be the optimal solution of the minimization problem
\begin{align}
    \label{prob1}
    \min_{\nu_1\leq\cdots\leq \nu_K}
    \sum_{k=1}^K \paren*{\Phi(\nu_k) - \nu_k s_k},
\end{align}
where $\Phi$ is a strictly convex function.
Then, the optimal solution $\hat{\mu} = (\hat{\mu}_1,\dots,\hat{\mu}_K)$ of the minimization problem
\begin{align}
    \label{prob2}
    \min_{\alpha \leq \mu_1\leq\cdots\leq \mu_K}
    \sum_{k=1}^K \paren*{\Phi(\mu_k) - \mu_k s_k}
\end{align}
is given by $\hat{\mu}_k = \max (\hat{\nu}_k, \alpha)$.
\end{theorem}

\begin{proof}
The Lagrangian function for \cref{prob1} is given by
\begin{align*}
    L(\nu,\eta)
    =
    \sum_{k=1}^K \paren*{\Phi (\nu_k)-\nu_k s_k }
    +
    \sum_{k=1}^{K-1} \eta_k (\nu_k - \nu_{k+1}).
\end{align*}
The optimal solution $\hat{\nu}$ and its corresponding multiplier $\hat{\eta}$ satisfy the KKT condition \cite{bv04}:
\begin{align}
    \label{KKTa1}
    \left. \frac{\partial F}{\partial \nu_k} \right|_{\nu=\hat{\nu},\: \eta=\hat{\eta}}
    =
    {\Phi}'(\hat{\nu}_k) - s_k + \hat{\eta}_k - \hat{\eta}_{k-1}=0, \quad k = 1,\dots,K,
\end{align}
where $\hat{\eta}_0 = \hat{\eta}_K = 0$ and 
\begin{align}
    \label{KKTa2}
    \hat{\eta}_k \geq 0,\quad 
    \hat{\nu}_k - \hat{\nu}_{k+1} \leq 0, \quad
    \hat{\eta}_k \paren*{\hat{\nu}_k - \hat{\nu}_{k+1} } = 0, \quad
    k = 1,\dots,K-1.
\end{align}

Similarly, the Lagrangian function for \cref{prob2} is defined as
\begin{align*}
    L(\mu,\lambda)
    =
    \sum_{k=1}^K \paren*{\Phi (\mu_k)-\mu_k s_k }
    +
    \sum_{k=0}^{K-1} \lambda_k (\mu_k - \mu_{k+1}),
\end{align*}
where $\mu_0 = \alpha$.
The KKT condition for the optimal solution $\hat{\mu}$ and its corresponding multiplier $\hat{\lambda}$ is given by
\begin{align}
    \label{KKTb1}
    \left. \frac{\partial F}{\partial \mu_k} \right|_{\mu=\hat{\mu},\: \lambda=\hat{\lambda}}
    =
    {\Phi}'(\hat{\mu}_k) - s_k + \hat{\lambda}_k - \hat{\lambda}_{k-1}=0, \quad k = 1,\dots,K,
\end{align}
where $\hat{\lambda}_K = 0$, $\hat{\mu}_0 = \alpha$ and 
\begin{align}
    \label{KKTb2}
    \hat{\lambda}_k \geq 0,\quad 
    \hat{\mu}_k - \hat{\mu}_{k+1} \leq 0, \quad
    \hat{\lambda}_k \paren*{\hat{\mu}_k - \hat{\mu}_{k+1} } = 0, \quad
    k = 0,\dots,K-1.
\end{align}
Note that the KKT condition \cref{KKTb1} and \cref{KKTb2} is not only necessary but also sufficient for the optimality of $\hat{\mu}$, since both the objective function and feasible region are convex \cite{bv04}.

Without loss of generality, assume
\begin{align*}
    \hat{\nu}_1 \leq \cdots \leq \hat{\nu}_{K^\prime}
    < \alpha
    \leq \hat{\nu}_{K^\prime+1} \leq \cdots \leq \hat{\nu}_K.
\end{align*}
Let $\hat{\mu}_k=\alpha$ for $k=1,\dots,K^\prime$ and $\hat{\mu}_k=\hat{\nu}_k$ for $k=K^\prime+1,\dots,K$.
Also, let $\hat{\lambda}_k = \hat{\eta}_k$ for $k=K^\prime,\dots,K-1$ and define $\hat{\lambda}_{K^\prime-1},\dots,\hat{\lambda}_0$ backwardly by $\hat{\lambda}_k=\Phi'(\hat{\mu}_{k+1})-s_{k+1}+\hat{\lambda}_{k+1}$.
We show that $(\hat{\mu},\hat{\lambda})$ satisfies \cref{KKTb1} and \cref{KKTb2}, which completes the proof.

First, from $\hat{\lambda}_k=\Phi'(\hat{\mu}_{k+1})-s_{k+1}+\hat{\lambda}_{k+1}$ for $k=0,\dots,K^\prime-1$, the condition \cref{KKTb1} is satisfied for $k=1,\dots,K^\prime$.
Also, since $\hat{\mu}_k=\hat{\nu}_k$ and $\hat{\lambda}_k = \hat{\eta}_k$ for $k=K^\prime+1,\dots,K$ and \cref{KKTa1} holds, the condition \cref{KKTb1} is satisfied for $k=K^\prime+1,\dots,K$.

From the third condition in \cref{KKTa2} for $k = K^\prime$ and the assumption $\hat{\nu}_{K^\prime} < \hat{\nu}_{K^\prime+1}$, we have $\hat{\lambda}_{K^\prime} = \hat{\eta}_{K^\prime} = 0$.
Since $\Phi$ is strictly convex, 
\begin{align*}
    \hat{\lambda}_{K^\prime-1} = \Phi'(\hat{\mu}_{K^\prime})-s_{K^\prime}+\hat{\lambda}_{K^\prime} \geq \Phi'(\hat{\nu}_{K^\prime})-s_{K^\prime}+\hat{\eta}_{K^\prime} = \hat{\eta}_{K^\prime-1}.
\end{align*}
By induction, we have $\hat{\lambda}_k \geq \hat{\eta}_k$ for $k=0,\dots,K^\prime-1$.
Thus, from the first condition in \cref{KKTa2}, the first condition in \cref{KKTb2} is satisfied for $k=0,\dots,K^\prime-1$.
Also, from $\hat{\lambda}_k = \hat{\eta}_k$ for $k=K^\prime,\dots,K-1$, the first condition in \cref{KKTb2} is satisfied for $k=K^\prime,\dots,K-1$.

From the definition of $\hat{\mu}$, the second condition in \cref{KKTb2} is satisfied for $k=0,\dots,K-1$.

Finally, from $\hat{\lambda}_{K^\prime}=0$ and $\hat{\mu}_k=\alpha$ for $k=0,\dots,K^\prime$, the third condition in \cref{KKTb2} are satisfied for $k=0,\dots,K^\prime$.
In addition, from $\hat{\mu}_k=\hat{\nu}_k$ and $\hat{\lambda}_k = \hat{\eta}_k$ for $k=K^\prime+1,\dots,K-1$ and the third condition in \cref{KKTa2}, the third condition in \cref{KKTb2} is also satisfied for $k=K^\prime+1,\dots,K-1$.
\end{proof}

\section{Geometric Integration and symplectic partitioned Runge--Kutta methods}
\label{app_gi}
Geometric numerical integration methods or structure-preserving numerical methods are numerical methods that preserve or inherit the underlying geometric properties of differential equations. 
The main advantage of geometric numerical integration methods is that in many cases we can expect qualitatively better numerical solutions, especially over a long period of time, than general-purpose methods.
For more details on this subject, see~\cite{hl03,hl06,lr04}.

Consider a coupled system
\begin{align}
    \label{eq:p_system}
    \frac{\rmd}{\rmd t}
    \begin{bmatrix}
    x \\ z
    \end{bmatrix}
    =
    \begin{bmatrix}
    f_1(x,z,t) \\
    f_2(x,z,t)
    \end{bmatrix},\quad
    \begin{bmatrix}
    x(0) \\ z(0)
    \end{bmatrix}
    =
    \begin{bmatrix}
    x_0 \\ z_0
    \end{bmatrix},
\end{align}
where $x\in\bbR^{M_1}$ and $z \in \bbR^{M_2}$ are time-dependent vectors and $f_1:\bbR^{M_1}\times\bbR^{M_2}\times\bbR\to\bbR^{M_1}$, $f_2:\bbR^{M_1}\times\bbR^{M_2}\times\bbR\to\bbR^{M_2}$. 
A partitioned Runge--Kutta (PRK) method applied to \cref{eq:p_system} reads
\begin{align*}
    x_{n+1} = x_n + \Delta t_n \sum_{i=1}^s b_i k_{n,i}, \quad 
    z_{n+1} = z_n + \Delta t_n \sum_{i=1}^s B_i l_{n,i}, 
\end{align*}
where
\begin{align*}
    k_{n,i} = f_1(X_{n,i},Z_{n,i},t_n + c_i \Delta t_n), \quad 
    l_{n,i} = f_2(X_{n,i},Z_{n,i},t_n + C_i \Delta t_n),
\end{align*}
and the internal stages $X_{n,i},Z_{n,i}$ for $i=1,\dots,s$ are defined by
\begin{align*}
    X_{n,i} = x_n + \Delta t_n \sum_{i=1}^s a_{ij}k_{n,j}, \quad
    Z_{n,i} = z_n + \Delta t_n \sum_{i=1}^s A_{ij}k_{n,j}.
\end{align*}
Note that in this appendix $t_n$ does note mean the time points the observations are made, but it simply means the time grid of the numerical method: $t_{n+1}-t_n = \Delta t_n$.

It is known that a PRK method preserves certain quadratic invariants if the coefficients of the method satisfies a certain condition. 
More precisely, the following theorem holds.

\begin{theorem}[e.g.~{\cite[Chapter~IV, Theorem~2.4]{hl06}}, {\cite[Theorem~2.4]{ss16}}]
Assume that $S:\bbR^{M_1}\times\bbR^{M_2}\to\bbR$ is a real valued bilinear map, and the solution to \cref{eq:p_system} satisfies
\begin{align*}
    \frac{\rmd}{\rmd t} S(x(t),z(t)) = 0.
\end{align*}
If the PRK coefficients satisfy
\begin{align}
    \label{symp_cond}
    b_i = B_i, \quad i=1,\dots,s,
    \quad
    b_{i}A_{ij} + B_j a_{ji} = b_iB_i,
    \quad i,j=1,\dots,s
\end{align}
and
\begin{align}
    \label{symp_cond_c}
    c_i = C_i, \quad i=1,\dots,s,
\end{align}
it follows that for the solution to the PRK method $S(x_n,z_n)$ is constant, i.e. it is independent of $n$.
\end{theorem}

A PRK method is called a symplectic PRK method if its coefficients satisfy \cref{symp_cond}, because such a PRK method exactly preserves the symplecticity when applied to Hamiltonian systems~\cite{hl06}.

The key idea to obtaining the exact gradient in \cref{subsec:update_theta} is to couple the original system \cref{eq:original1} and the adjoint system, and to apply a symplectic PRK method
(precisely speaking, the variational equations $\dot{\delta} = \nabla _x f(x) \delta$ should be taken into account, but we omit the detail since in practice there is no need to integrate the variational equations).
For example, when $s=1$, if the original system is integrated by using the explicit Euler method \cref{eq:exEuler}, i.e. $(a,b,c)=(1,1,1)$, then the adjoint system should be solved by the method $(A,B,C)=(0,1,1)$ so that the pairs of the coefficients satisfy the conditions \cref{symp_cond} and \cref{symp_cond_c}.
The choice $(A,B,C)=(0,1,1)$ leads to the scheme \cref{eq:Euler_sym}.

\section{The Störmer--Verlet method}
\label{app_sv}

The Störmer--Verlet method is a typical symplectic partitioned Runge--Kutta method \cite{hl06}.
It has been widely used to integrate Hamiltonian systems, especially those from celestial mechanics and molecular dynamics.
We employ the Störmer--Verlet method to integrate the Kepler equation in \cref{subsec:Kepler}.

In general, for a system of ODEs of the form
\begin{align} 
    \label{eq:cm}
    \frac{\rmd}{\rmd t} 
    \begin{bmatrix}
        q \\ p 
    \end{bmatrix}
    = 
    \begin{bmatrix}
        p \\ -f(q)
    \end{bmatrix},
\end{align}
the one step formula $(\tq_n,\tp_n)\mapsto (\tq_{n+1},\tp_{n+1})$ of the Störmer--Verlet method is defined by
\begin{align*}
    \tq_{n+1/2} &= \tq_n + \frac{\Delta t_n}{2} \tp_n,\\
    \tp_{n+1} &= \tp_n - \Delta t_n f(\tq_{n+1/2}) ,\\
    \tq_{n+1} &= \tq_{n+1/2} + \frac{\Delta t_n}{2} \tp_{n+1},
\end{align*}
where $\Delta t_n$ is the time step size.

The adjoint system \cref{eq:adjoint1} for \cref{eq:cm} is
\begin{align}
    \label{eq:adj_sv}
    \frac{\rmd}{\rmd t}
    \begin{bmatrix}
        \lambda \\ \nu
    \end{bmatrix}
    =
    \begin{bmatrix}
        \nabla _q f(q) ^\top \nu \\
        -\lambda
    \end{bmatrix}.
\end{align}
To obtain the exact gradient defined with the numerical solutions of the Störmer--Verlet method, we need to follow the discussion in \cref{app_gi}.
But the difficulty is that while in \cref{app_gi} (and~\cite{ss16}) the original equation is assumed to be solved by a standard Runge--Kutta method, in this case the original equation is solved by the PRK method.
It turns out that the discussion in \cref{app_gi} can be extended to the case that the original equation itself is solved by a PRK method.
To obtain the exact gradient defined with the numerical solutions of the Störmer--Verlet method,
the adjoint system \cref{eq:adj_sv} should be numerically integrated backwardly by the formula
\begin{align*}
    \tnu_{n+1/2}  &= \tnu_n -\frac{\Delta t_n}{2} \tlambda_n \\
    \tlambda_{n+1} &=  \tlambda_n +  \Delta t_n \nabla _q f(\tq_{n+1/2})^\top \tnu_{n+1/2}, \\
    \tnu_{n+1} &= \tnu_{n+1/2}-\frac{\Delta t_n}{2} \tlambda_{n+1}.
\end{align*}

\bibliographystyle{siamplain}
\bibliography{references}

\end{document}